\documentclass[sigconf, nonacm]{acmart}

\newcommand\vldbdoi{XX.XX/XXX.XX}
\newcommand\vldbpages{XXX-XXX}
\newcommand\vldbvolume{14}
\newcommand\vldbissue{1}
\newcommand\vldbyear{2020}
\newcommand\vldbauthors{\authors}
\newcommand\vldbtitle{\shorttitle} 
\newcommand\vldbavailabilityurl{https://anonymous.4open.science/r/GLORAN-56FF}
\renewcommand\footnotetextcopyrightpermission[1]{}
\newcommand\vldbpagestyle{plain}

\usepackage{tikz}
\usepackage{amsmath}

\usepackage{amssymb}
\usepackage[ruled,linesnumbered]{algorithm2e}
\usepackage[cal=boondox]{mathalfa}
\usepackage{multirow}
\usepackage{subcaption}

\usepackage[normalem]{ulem}
\useunder{\uline}{\ul}{}

\usepackage{amsthm, amsmath, amssymb}
\newtheorem{problem}{Problem}

\newcommand\ourmethod{GLORAN}
\newcommand\decom{\it Decomp}
\newcommand\scan{\it Scan\&D}
\newcommand\lookup{\it Lookup\&D}
\newcommand\rocksdb{\it RocksDB}

\begin{document}
\title{Don't Forget Range Delete! Enhancing LSM-based Key-Value Stores with More Compatible Lookups and Deletes}

\author{Fan Wang}
\email{fan008@e.ntu.edu.sg}
\affiliation{%
  \institution{Nanyang Technological University}
  \country{Singapore}
}

\author{Dingheng Mo}
\email{dingheng001@e.ntu.edu.sg}
\affiliation{%
  \institution{Nanyang Technological University}
  \country{Singapore}
}

\author{Siqiang Luo}
\email{siqiang.luo@ntu.edu.sg}
\affiliation{%
  \institution{Nanyang Technological University}
  \country{Singapore}
}

\begin{abstract}
LSM-trees are featured by {\it out-of-place} updates, where key deletion is handled by inserting a tombstone to mark its staleness instead of removing it in place. This defers actual removal to compactions with greatly reduced overhead. However, this classic strategy struggles with another fundamental operator--range deletes--which removes all keys within a specified range, requiring the system to insert numerous tombstones and causing severe performance issues.

To address this, modern LSM-based systems introduce range tombstones that record the start and end keys to avoid per-key tombstones. Although this strategy achieves impressive range delete efficiency, we show that such de facto industrial range delete solution is {\it incompatible} with lookup operators. 
In particular, our experiments show that point lookup latency can increase by 30\% even with just 1\% range deletions in the workload. Further to our surprise is that this issue has not been raised before, although the range tombstone solution has been employed for more than five years.

To address this critical system performance issue, we propose {\ourmethod}, an efficient range delete method that can be integrated into modern LSM-based systems and offers desirable range deletion performance without compromising point lookup efficiency. It introduces a global index to efficiently manage deleted ranges. This index allows point lookups to quickly locate the relevant deleted ranges without retrieving many irrelevant elements, reducing the I/O complexity from $O(\frac{N}{\lambda})$ to either $O( \log^2 \frac{N}{\lambda F})$ or $O(\varphi \cdot \log{\frac{N}{F}})$, where $1/\lambda$ is the ratio of range deletes, and $\varphi$ is the FPR of Bloom filters in LSM-trees. 
Furthermore, we design an entry validity estimator to further enhance expected I/O cost to {$O(\varepsilon \cdot \log^2{\frac{N}{\lambda F}})$} for looking up existing keys. 
Extensive evaluations demonstrate that {\ourmethod} consistently outperforms baseline approaches, while achieving up to $10.6\times$ faster point lookups and 
$2.7\times$ higher overall throughput compared to the state-of-the-art method.

\end{abstract}

\maketitle

\pagestyle{\vldbpagestyle}
\begingroup\small\noindent\raggedright\textbf{PVLDB Reference Format:}\\
\vldbauthors. \vldbtitle. PVLDB, \vldbvolume(\vldbissue): \vldbpages, \vldbyear.\\
\href{https://doi.org/\vldbdoi}{doi:\vldbdoi}
\endgroup
\begingroup
\renewcommand\thefootnote{}\footnote{\noindent
This work is licensed under the Creative Commons BY-NC-ND 4.0 International License. Copyright is held by the owner/author(s). Publication rights licensed to the VLDB Endowment. \\
\raggedright Proceedings of the VLDB Endowment, Vol. \vldbvolume, No. \vldbissue\ %
ISSN 2150-8097. \\
\href{https://doi.org/\vldbdoi}{doi:\vldbdoi} \\
}\addtocounter{footnote}{-1}\endgroup

\ifdefempty{\vldbavailabilityurl}{}{
\vspace{.3cm}
\begingroup\small\noindent\raggedright\textbf{PVLDB Artifact Availability:}\\
The source code, data, and/or other artifacts have been made available at \url{\vldbavailabilityurl}.
\endgroup
}

\section{Introduction}
The Log-Structured Merge-tree (LSM-tree) has been widely employed as the backbone of many renowned key-value stores such as RocksDB~\cite{rocksdb}, Cassandra~\cite{lakshman2010cassandra}, AsterixDB~\cite{alsubaiee2014asterixdb}, ScyllaDB~\cite{scylladb}, CockroachDB~\cite{cockroachdb}, and InfluxDB~\cite{influxdb}. It organizes data as key–value pairs, or {\it entries}, across multiple levels with exponentially increasing capacities. The smallest level, $L_0$, resides in memory, while all other levels are stored on disk, with entries sorted by key. New entries are first inserted into $L_0$ and progressively migrated to larger levels through compaction processes. The LSM-tree supports efficient key deletion through an {\it out-of-place} approach. Instead of finding the key and removing it in place, it inserts a special entry, a {\it tombstone}, to conceptually mark the key as invalid. The obsolete entries remain until purged by the tombstone during subsequent compactions as Figure ~\ref{fig:point_delete} shows. This efficiently handles single key deletions, while, still falling short in supporting range deletes.

\begin{figure}[t]
\centering
  \setlength{\abovecaptionskip}{0 cm}
  \includegraphics[width=0.99\linewidth]
  {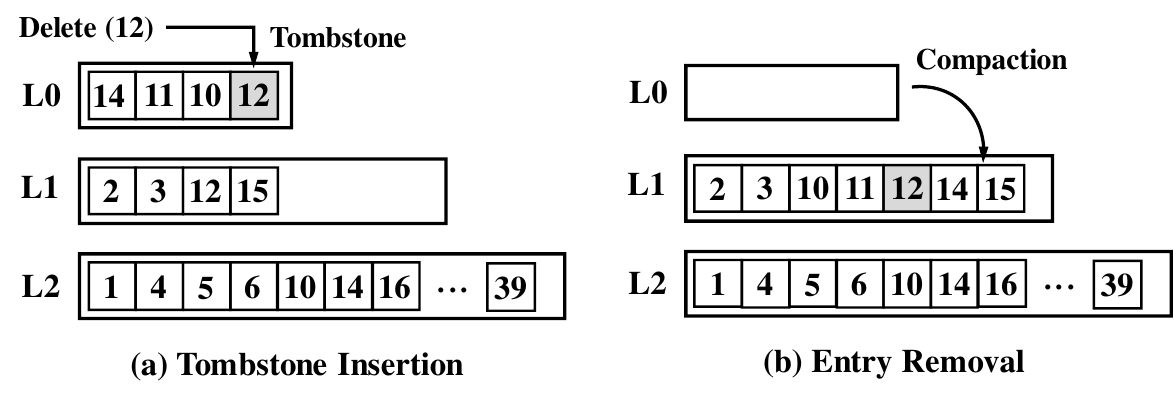} 
  \caption{The LSM-tree employs out-of-place deletion with tombstones, and removes the entry during compaction.}
\label{fig:point_delete}
\vspace{-4mm}
\end{figure}

A range delete removes all keys within a certain range from the LSM-tree. This operation is naturally utilized in modern applications and has been incorporated by many renowned LSM-based key-value stores like RocksDB and Cassandra. For example, in e-commerce platforms, time-limited promotional activities may assign a special prefix to product codes, grouping all promotional items within a specific key range. Once the activity ends, removing the associated data requires efficient range delete. Similar needs arise in other applications, such as purging time-bound data in time series databases~\cite{influxdelete} or discarding outdated dataset versions in machine learning pipelines. 

\vspace{1mm}
\noindent
{\bf Surprising Impact of Range Delete.} While LSM-tree optimizations have been extensively studied, the impact of range delete remains underexplored. While one may not expect a heavy portion of range deletes in key-value store workloads, to our surprise, our benchmark shows that even a tiny range delete portion (e.g., 1\%) can significantly degrade the performance of {\it other operations} (e.g., {\it point lookup} that retrieves the entry for a specific key). This calls for rethinking the system design regarding range deletes.
Naively, implementing range delete in an LSM-tree requires inserting a tombstone for each key within the target range. This incurs substantial insertion overhead and additional costs for the subsequent compactions, especially for long ranges. To mitigate this, the state-of-the-art method introduces {\it range tombstone}. It represents an entire key range with a single entry by encoding both the start and end keys. With this design, only one entry is inserted per range delete, thus significantly reducing the I/O cost. As a result, this method is widely adopted in modern LSM-based key-value stores such as RocksDB and CockroachDB. While this approach delivers outstanding performance for range deletes, it may significantly drag down the point lookup performance. 
In our experiments, {\it in a workload without range deletes, replacing just 1\% of operations with range deletes increased point lookup latency by more than 30\% in RocksDB. The point lookup performance can even drop by 2.3 times if the range delete workload occupies 5\%.} Despite such significant impact, the issue has not been explicitly studied. In this paper, we conduct a thorough analysis of this issue and propose a novel approach that effectively improves point lookup efficiency while delivering desirable range delete performance.

The state-of-the-art design introduces a dedicated block at each LSM-tree level to accommodate the range tombstones. This local storage layout leads to substantial and often unnecessary overhead for point lookups. In order to find a specific key, the system searches the LSM-tree levels from top to bottom. At each level, the corresponding range tombstone block must also be checked to determine whether the key has been deleted. Many levels, however, do not contain the target key, leading to unnecessary accesses of the range tombstone blocks. Furthermore, within a certain block, point lookups would retrieve numerous unrelated range tombstones, incurring high I/O costs. Although range tombstones are sorted by their start key, their variable lengths prevent the system from knowing the exact key range coverage before retrieving them. As a result, every tombstone with a start key smaller than the target key must be examined, which greatly increases lookup cost.

To address these inefficiencies effectively, we propose managing the deleted ranges within a global index. 
To enhance clarity, we refer to the deleted range issued by a range delete operation as its range record. The global index captures the coverage information of range records and organizes them accordingly. Unlike the state-of-the-art method that indexes ranges solely by their start key, our approach enables efficient locating of the relevant ranges for a target key while skipping large amounts of irrelevant ones.
Moreover, the global index can be bypassed when the target key is absent from the LSM-tree, thereby reducing the overhead.
Hence, this design leads to obviously improved point lookup efficiency while retaining desirable range delete support.

Constructing such a global index is by no means straightforward. If range tombstones are stored naively, their temporal information is lost, which may cause incorrect deletion of keys that arrive after the range deletion. To solve this problem, we introduce a two-dimensional representation for range records. Each record explicitly encodes both the key boundaries and the time boundaries of a deletion. 
We then present an efficient spatial structure, LSM-DRtree, as the global structure for managing range records. By properly exploiting the intrinsic spatial relationships among range records, this structure effectively eliminates overlaps within the index, thus achieving impressive point lookup performance.
Our theoretical analysis shows that, compared with the state-of-the-art design, the point lookup cost can be reduced from $O(\frac{N}{\lambda})$ to $O( \log^2 \frac{N}{\lambda F})$, where $\frac{N}{\lambda}$ is the number of range deletes.
At the same time, the LSM-style structure ensures its update efficiency that satisfies the performance of range deletes.
It also enables the index to operate in minimal memory overhead, with the buffer size limited to 4 MB in our experiments (and one can foresee performance upgrade when more buffer memory is allowed).
Beyond the global index, we enhance the design with an entry validity estimator to further speed up point lookups.
We summarize our contributions as follows.

\begin{itemize}
    \item 
    To our knowledge, the impact of range delete has not been comprehensively reviewed in the literature. We provide a first comprehensive quantitative analysis of the range delete overhead and its impact on other operations under various existing methods.
    
    \item 
    We propose {\ourmethod}, a novel range delete strategy that organizes range records globally, achieving both efficient range delete and high-performance lookups.

    \item 
    We design an efficient global index structure, LSM-DRtree, which combines fast range deletions with low-latency point lookups. Together with a lightweight entry validity estimator, it reaches further enhanced lookup performance.

    \item 
    We conduct extensive evaluations under diverse workloads, where {\ourmethod} consistently delivers competitive performance, achieving up to $2.7\times$ higher throughput than state-of-the-art methods.
    

\end{itemize}

\begin{table}[t]
\begin{tabular}{l|l|l}
\hline
Term & Definition                & Unit \\ \hline
$F$    & Capacity of memory buffer of an LSM-tree   &entries   \\
$B$    & Size of data blocks                    &bytes   \\
$k$    & Size of key in an LSM-tree                &bytes   \\
$e$    & Size of entry in an LSM-tree              &bytes   \\
$\varphi$    & False positive rate of Bloom filter              &   \\
$\epsilon$    & False positive rate of extendable validity filter              &   \\
$N$    & Number of entries in an LSM-tree          &        \\
$Q$    & Number of range deletes                &        \\
$\ell$    & Average length of deleted range        &        \\
$L$    & Number of levels in an LSM-tree         &      \\
$T$    & Size ratio of an LSM-tree         &      \\
\hline
\end{tabular}
\vspace{-1mm}
\caption{List of terms used throughout the paper. 
}
\vspace{-10mm}
\label{tab:terms}
\end{table}

\vspace{-2mm}
\section{Background}
\label{sec:bg}
This section introduces the basic operations of LSM-trees and analyzes their associated costs under workloads without range deletes. This prepares readers for the next section, where we discuss how the range deletes impact the workflow and cost of these operations.
Following prior works~\cite{dostoevsky2018,huynh2021endure,liu2024structural}, the operational costs are evaluated as the number of I/Os, with terms used throughout the paper summarized in Table~\ref{tab:terms} for ease of reference.

\vspace{1mm}
\noindent {\bf Update the LSM-tree.}
The LSM-tree consists of multiple levels, among which $L_0$ resides in the main memory serving as a write buffer with the capability of accommodating $F$ entries. For the on-disk levels, their capacity is $T$ times larger than their upper level, where $T$ is the size ratio. In other words, the capacity of the $i$-th level is $F \cdot T^i$. 
Therefore, to store $N$ entries, $L = \log_T(\tfrac{N}{F})$ levels should be included.
To update the LSM-tree, a new entry is inserted into the write buffer that can eventually be compacted to the bottommost level. Therefore, an entry can be rewritten up to $O(T \cdot L)$ times. The entries are processed by sequential read/write that handles data in the granularity of disk data blocks with size of $B$. Hence, the update cost is $O(T \cdot L \cdot \tfrac{e}{B})$ I/Os, where $e$ is the size of an entry. 
Here, we adopt a typical LSM design, leveling configuration, in our analysis for clarity and readability. We also provide associated analysis on other configurations in our technical report~\cite{technicalreport}.

\vspace{1mm}
\noindent {\bf Point Lookup.}
This operation retrieves the entry with a target key from the LSM-tree through a level-by-level search. To accelerate this process, modern LSM-trees introduce a fence pointer for each data block that records the associated key range. Since the entries are sorted in the disk levels, this facilitates identifying the data block that may contain the target key in each level and retrieving it with a single I/O. In addition, the Bloom filters further reduce I/O cost by predicting the presence of a matched key in a certain level without false negative results. Hence, we can safely skip the levels when the Bloom filter returns false to save I/O. We denote the false positive rate of a Bloom filter as $\varphi$. Then, if there is no matched key in the LSM-tree, the operation should look through all levels with a cost of $O(\varphi \cdot L)$. If the matched key does exist, it requires an additional I/O to retrieve the entry, counting $O(1 + \varphi \cdot L)=O(\lceil \varphi \cdot L \rceil)$.

\vspace{1mm}
\noindent {\bf Point Delete.}
This operation removes a specific key from the LSM-tree using an out-of-place strategy. Instead of directly locating and immediately deleting obsolete entries, the system inserts a {\it tombstone} into the memory buffer to mark the target key as invalid. A tombstone is a special entry that carries only the key without an associated value. During flush or compaction, it is processed like a normal entry and used to discard obsolete entries containing the same key. Moreover, since a key can have multiple outdated versions across levels, the tombstone must be retained until its expiration when reaching the bottommost level, which incurs an I/O cost of $O(T \cdot L \cdot \tfrac{k}{B})$. Besides, point lookups can stop at tombstones as they confirm the non-existence of target keys.

\begin{table}[t]
\begin{tabular}{l l l l l l}
\hline
    {\bf Operation} & {\bf LRR (SOTA)} & {\bf {\ourmethod} (ours)} \\
\hline
RangeDelete
& $O\! \left(\tfrac{k}{B} \cdot T \cdot \log{\tfrac{N}{F}}\right)$ 
& $O\! \left(\tfrac{k}{B} \cdot T \cdot \log\left( \tfrac{1}{\lambda} \cdot \tfrac{N}{F} \right)\right)$  \\

Lookup (V)  
&  $O\! \left(\tfrac{N k}{\lambda B} \! + \! \lceil \varphi \rceil \log \! \tfrac{N}{F} \! + \! 1 \right)$
&  $O\!\left(\varepsilon  \log^2 \tfrac{N}{\lambda F} \! + \! \lceil \varphi \log \! \tfrac{N}{F} \rceil \right)$ \\

Lookup (N)  
& $O\! \left(\tfrac{N k}{\lambda B} \! + \! \lceil \varphi \rceil \log \! \tfrac{N}{F} \right)$
& $O\! \left(\varphi \log{\tfrac{N}{F}} \right)$  \\

Lookup (O)  
& $O\! \left(\tfrac{N k}{\lambda B} \! + \! \lceil \varphi \rceil \log \! \tfrac{N}{F} \right)$  
& $O\! \left(\log^2{\tfrac{N}{\lambda F}} \! + \! \lceil \varphi \log{\tfrac{N}{F}} \rceil \right)$  \\
\hline
\end{tabular}
\caption{Operational cost of different range delete methods, where V, N, and O represent lookups on keys that are valid, non-existent in the LSM-tree, and obsoleted by range deletes while not removed from LSM-tree. In this table, all logarithmic terms are taken with base T. {\ourmethod} delivers better lookup performance than the SOTA method across diverse cases that can reduce the cost from $O\left(\tfrac{N}{\lambda}\right)$ to $O\left(\log^2 \left(\tfrac{N}{\lambda F}\right)\right)$. It also enables enhanced range delete, which matches the performance of LLR with only $F/16$ memory assigned to the global index in our experiments.}
\vspace{-7mm}
\label{tab:costs}
\end{table}
\section{The impact of Range Deletes}
\label{sec:range delete impact}
While the previous analyses give hints on the performance of point-lookup/delete and update operations, these analyses become inaccurate with the consideration of range deletes. 
Hence, this section further presents an analysis of the associated costs {\it in the presence of range deletes}, which have not been formally studied before. 
To this end, we detail the workflow of relevant operations and quantitatively evaluates their cost as summarized in Table~\ref{tab:costs}. These findings not only offer a deeper understanding of the range delete operation but also inspire the design of our global range delete method.

Modern key–value stores introduce a specialized entry, range tombstone, to represent the range records. For each deleted range, the range tombstone records the start key as the entry key, while storing the end key in its value. This design allows a compact representation of range records, which remains consistent across various range delete lengths. For instance, as the key size is $k$, a range tombstone only takes up $2k$ in size. 
To store these range tombstones, a dedicated block is allocated in each LSM level. Within these blocks, the range tombstones are sorted by their start keys to improve query efficiency. Since these tombstones are locally stored at individual levels, we refer to this approach as the local range record (LRR) method.

\vspace{1mm}
\noindent {\bf Range Delete Operation.}
The range delete can be performed by simply inserting a range tombstone into the range tombstone block of the memory buffer. These blocks would be consulted during flushes or compactions to remove obsolete entries. Similar to point deletions, range tombstones must be propagated to the bottommost level to ensure that all obsolete entries are eventually cleared. To achieve this, range tombstone blocks are also merged during flushes and compactions, with the results written into the target level. This process requires reading and rewriting the range tombstones, incurring an I/O overhead of $O(\frac{k}{B})$ per entity. As analyzed in Section~\ref{sec:bg}, a range tombstone typically undergoes $O(T \cdot L)$ compactions before expiration, leading to a total range delete cost of $O(T \cdot L \cdot \frac{k}{B})$. 
This result highlights the notable efficiency of the LRR approach, as its cost is even lower than that of a single update operation. 

\vspace{1mm}
\noindent {\bf The Impact on Point Lookups.}
Since range tombstone blocks are locally attached to each LSM-tree level, point lookups still need to search the tree from the memory buffer down to the bottommost level. Whereas, at each level, in addition to checking the data blocks for the target entry, the corresponding range tombstone block must also be probed to determine whether the key has been deleted. Then, the lookup can terminate if the key is found in the data blocks and not invalidated by any tombstones, or if it is determined to be deleted by a range tombstone. Therefore, apart from the lookup cost analyzed in Section~\ref{sec:bg}, probing range tombstone blocks brings extra overhead, which is evaluated in detail in the following part.


Within a certain range tombstone block, a target key $v$ can be covered by any range tombstone whose key is smaller than $v$ due to the various range lengths. For instance, key $500$ can be deleted by either range delete $[1,1000)$ or $[490,510)$. To estimate the number of such tombstones to check, let $\frac{N_i}{\lambda}$ denote the number of range tombstones in the $i$-th level, where $N_i$ and $1/\lambda$ represent the number of data entries and the ratio of range deletes, respectively. The key of the range tombstones is uniformly distributed over the key universe $[0,U)$. In this case, the number of candidates $q_i$ follows a binomial distribution
with the expected value of $\frac{N_i}{2\lambda}$. In order to retrieve them, one I/O is required to load the first page of the tombstone block, followed by sequential reads since tombstones are sorted by start key. Thus, the probing cost is $O(1 + \frac{N_i}{\lambda} \cdot \frac{k}{B})$, together with $\varphi$ I/Os to search the data blocks. If $v$ does not exist or is invalidated by a tombstone, all $L$ levels must be checked with cost presented in Equation ~\ref{eq:lrr_pointlookup_no}. If a valid $v$ exists, one more I/O is needed to retrieve it.

\vspace{-3mm}
{\small
\begin{align}
Z &= O\left( \sum_{i=1}^{L}\left(\frac{N_i}{\lambda} \cdot \frac{k}{B} + \varphi + 1\right)\right)
= O\left(\frac{N}{\lambda} \cdot \frac{k}{B} + L \cdot \varphi + L \right)
\label{eq:lrr_pointlookup_no}
\end{align}
}


\vspace{-1mm}
\noindent {\bf Discussion.}
As indicated in Equation~\ref{eq:lrr_pointlookup_no}, the range tombstone blocks incur at least $L$ additional I/Os during point lookups. This overhead is considerable compared with the cost of data block access, since the false positive rate of Bloom filter $\varphi$ is usually much smaller than 1. The key reason is that range tombstone blocks are stored locally, hence requiring at least one I/O at every level to retrieve them. Consequently, the overhead persists even in workloads with only a few range deletes. In addition, the lack of an efficient index forces point lookups to examine $O(\frac{N}{\lambda})$ tombstones. Within a specific level, all range tombstones whose start key is smaller than the target are retrieved. Whereas many of them are irrelevant to the target key. For instance, when searching the key $100$, the range tombstone $(5, 25)$ should be accessed, though it does not overlap with $100$, which leads to unnecessary I/Os.
These problems can be effectively mitigated by introducing an efficient global index for range records, which decouples the probing cost from the number of LSM levels. Furthermore, organizing records by their ranges allows the system to skip many irrelevant ones, thus achieving lower I/O cost. 
Based on these analyses, we propose a novel method that can significantly reduce point lookup cost 
while maintaining efficient range deletes.

\begin{figure}[t]
\vspace{-6mm}
\centering
  \setlength{\abovecaptionskip}{0 cm}
  \includegraphics[width=0.99\linewidth]
  {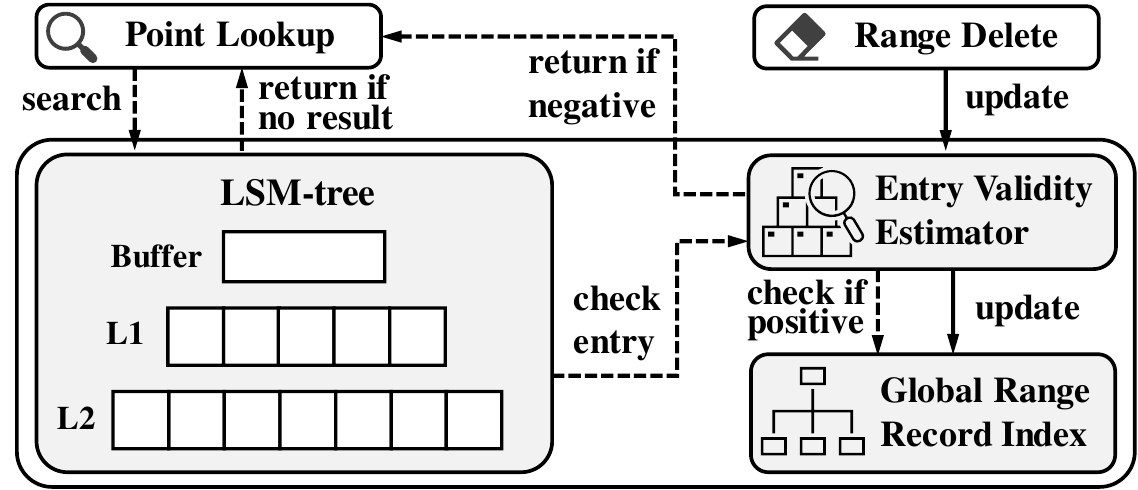} 
  \caption{An illustration of the GLORAN structure.}
\label{fig:gloran structure}
\vspace{-5mm}
\end{figure}

\section{\ourmethod}
\label{sec:global index}
We introduce a global range delete method, {\ourmethod}, as illustrated in Figure~\ref{fig:gloran structure}. It incorporates a global index to organize the range records instead of storing them in LSM levels. Under this new configuration, the conventional range tombstone is no longer suitable. To tackle this, we design a novel representation approach to properly profile the range records (Section~\ref{subsec:effective_area}). Besides, the associated index structure, LSM-DRtree, is proposed to manage them efficiently. This structure delivers impressive query performance thus providing strong efficiency guarantees for point lookups while preserving desirable range delete performance (Section~\ref{subsec:lsm-rtree}). Moreover, {\ourmethod} employs a lightweight entry validity estimator that offers a shortcut in point lookup workflow, which further reduces expected lookup overhead (Section~\ref{subsec:filter}).


\vspace{-2mm}
\subsection{Effective Area of Range Records}
\label{subsec:effective_area}
The global range delete method necessitates effective approach to represent the range records for efficient management. This, however, is nontrivial. Due to the out-of-place deletion strategy of LSM-tree, simply recording the key range of a range delete is insufficient because each record implicitly carries temporal information. For example, suppose a range delete inserts a range $[5,15)$ in the global index, followed by an update on key 8. Then, the global index would still mark this key as invalid, thus leading to incorrect point lookup results and mistaken removal of valid entries during compaction. 
Consequently, an effective representation of range records must capture both the key range and the temporal information. This can be naturally achieved when range records are stored locally in each LSM level within range tombstone blocks. As a special entry type, range tombstones are automatically assigned sequence numbers and managed by the multi-version control mechanism of the LSM-tree. However, in the global index, range records are stored outside the LSM-tree, making this mechanism unachievable. Therefore, a novel representation is required to explicitly encode both the key range and the sequence number range of each record.

\begin{figure}[t]
\vspace{-6mm}
\centering
  \setlength{\abovecaptionskip}{0 cm}
  \includegraphics[width=0.99\linewidth]
  {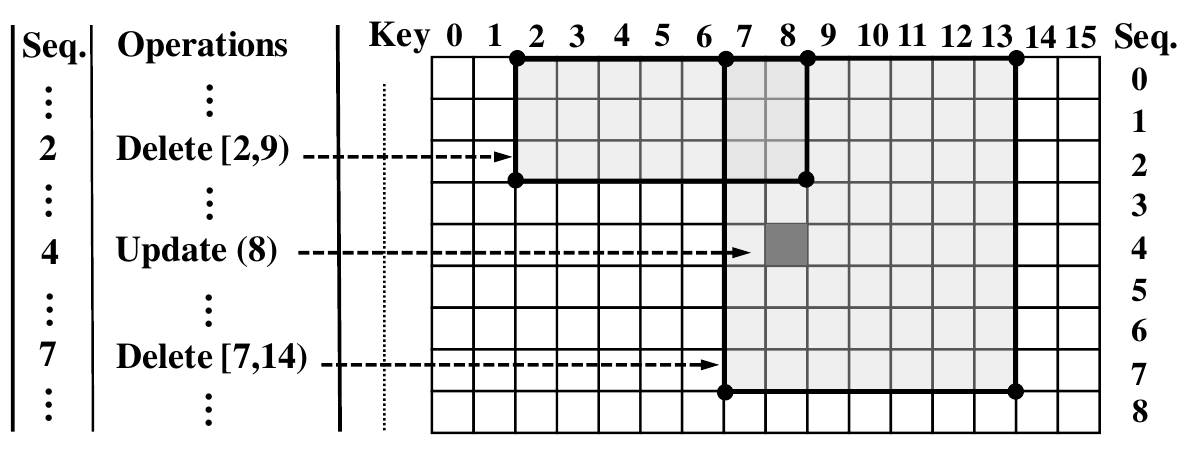} 
  \caption{The rectangular effective area of a range delete within working space invalidates all entries within it.}
\label{fig:effective area}
\vspace{-5mm}
\end{figure}

To this end, we introduce a two-dimensional working space spanning the key domain and the sequence number domain. 
As shown in Figure~\ref{fig:effective area}, a range delete on keys $[7, 14)$ is issued with sequence number 8. This operation invalidates all entries whose keys fall within the specified range and sequence numbers are smaller than 8 in the LSM-tree. These conditions naturally enclose a rectangle in the working space, which we define as its {\bf effective area}. Then, any entry mapped inside this area should be obsoleted by this operation, such as the entry with key 8 and sequence number 5 in Figure~\ref{fig:effective area}. Based on this abstraction, the validity of entries can be determined according to Lemma~\ref{lemma: effective area}.

\begin{lemma}
For a sequence of range deletes with effective areas $\mathcal{A} = \{\alpha_1, \alpha_2, \cdots , \alpha_n \}$, any entry invalidated by a range delete must be covered by a ${\alpha_i \in \mathcal{A}}$ in the working space.
\label{lemma: effective area}
\end{lemma}

Accordingly, the global index represents each range record by its effective area, a rectangle defined by the two vertices. For example, in Figure~\ref{fig:effective area}, a range delete on keys $[7, 14)$ with sequence number 8 is represented by two vertices: $(7, 0)$, denoting the smallest key and sequence number, and $(14, 8)$, their largest counterparts. Note that the smallest sequence of a range record specifies the minimum boundary before which the record expires, which differs across range records. Figure~\ref{fig:effective area} shows it as zero for simplicity. 
In general, the effective area can precisely capture the boundaries of each range record to avoid unnecessary lookups for irrelevant entries. Observant readers may notice that an appropriate index over these rectangles is also crucial for the system performance. We thus next introduce the design of our global index structure.


\begin{figure}[t]
\vspace{-6mm}
\centering
  \setlength{\abovecaptionskip}{0 cm}
  \includegraphics[width=0.99\linewidth]
  {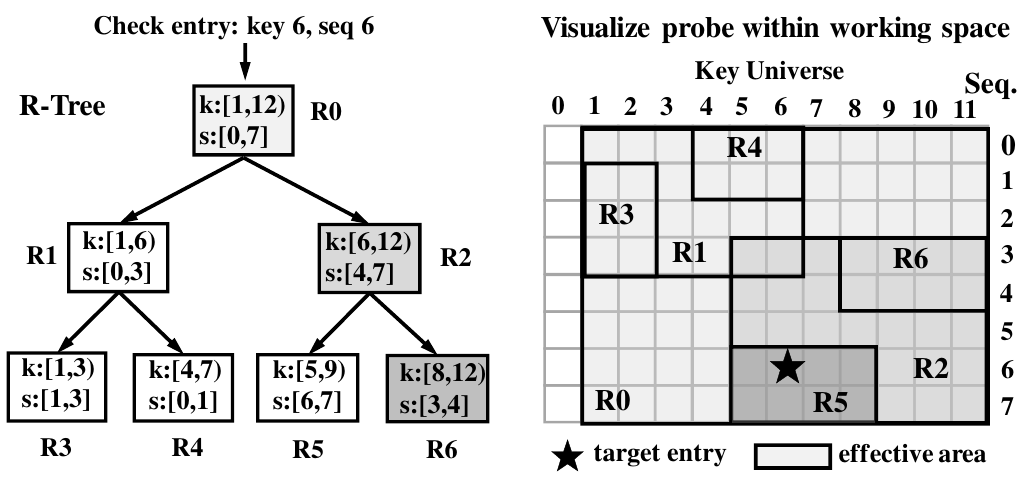} 
  \caption{The R-tree structure is able to locate an entry efficiently within the working space during point lookups.}
\label{fig:lsm rtree}
\vspace{-4mm}
\end{figure}

\subsection{Global Range Record Index: LSM-DRtree}
\label{subsec:lsm-rtree}
The global index is utilized to store range records and determine whether a key has been invalidated by them, which is key for achieving desirable range deletes and point lookups. Considering the rectangular effective areas of the range records, the index must efficiently store a set of rectangles and determine whether a given point is covered by any of them.

In this context, the performant spatial R-tree serves as a natural and attractive choice.
As illustrated in Figure~\ref{fig:lsm rtree}, an R-tree organizes range records in a hierarchical manner: leaf nodes (R3–R6) store the actual range records, while internal nodes (R1, R2) maintain the minimum bounding rectangles (MBRs) that enclose all their child nodes. Higher-level internal nodes recursively index lower ones until a single root node remains. To check whether an entry is covered by any range record, the R-tree is probed from the root downward. For instance, in Figure~\ref{fig:lsm rtree}, checking an entry with key 6 and sequence number 6 sequentially visits nodes R0, R2, and R6. The effective areas of these nodes are also presented for clarity. During this descent, the search region narrows progressively, thus effectively skipping irrelevant elements. When inserting a new range record, the R-tree attaches it to the internal node whose MBR incurs the smallest expansion, which naturally clusters related records and minimizes MBR overlap, thus improving query efficiency.

Despite its practical effectiveness, the R-tree suffers from theoretical inefficiency due to its lack of sound worst-case query complexity guarantees. When range records are highly skewed in the key universe, their effective areas become densely clustered, resulting in significant MBR overlap among internal nodes. Consequently, certain queries may need to access substantial, even nearly all, internal nodes in each level with excessive overhead. 
For example, if the queried entry in Figure~\ref{fig:lsm rtree} has a sequence number of 3, R1 cannot be skipped even though none of its range records actually cover the entry. This property results in undesirable tail latency and unstable performance.
Moreover, the R-tree incurs heavy update overheads. As the number of range records grows, a node may exceed its capacity, which necessitates a split and adjustments across the parent nodes. These operations are complex and costly. In particular, the global index must be serialized on disk to accommodate large volumes of range records. Hence, direct updates to on-disk R-tree nodes lead to substantial I/O overhead that significantly degrades range delete performance of the system.


To address these issues, LSM-DRtree is proposed to enhance the theoretical query complexity of R-tree using the disjoint property of effective areas, as well as improve update performance with an LSM-style structure, as detailed in the following.


\begin{figure}[t]
\vspace{-6mm}
\centering
  \setlength{\abovecaptionskip}{0 cm}
  \hspace{-4mm}
  \includegraphics[width=0.99\linewidth]
  {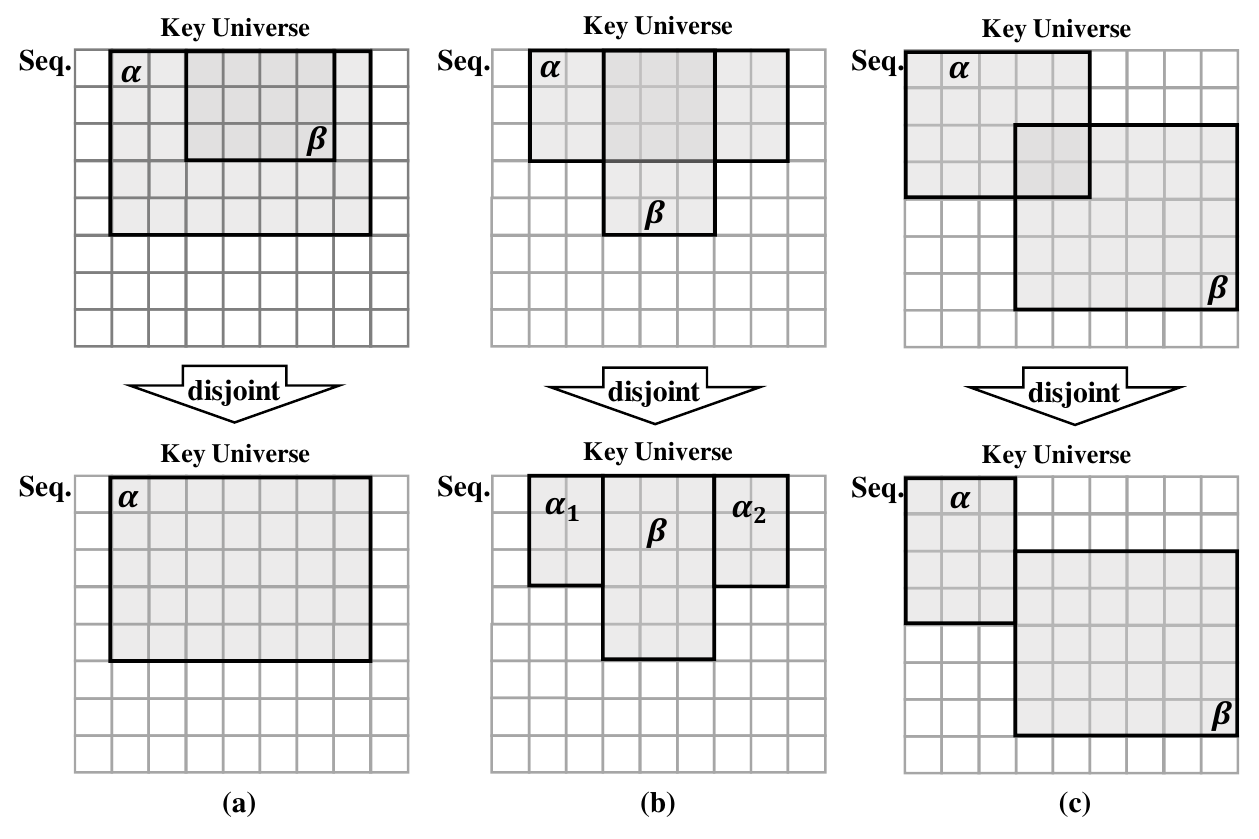} 
  \caption{Disjointization between two effective areas.}
\label{fig:effective-area-merge}
\vspace{-5mm}
\end{figure}

\vspace{1mm}
\noindent
{\bf Effective Areas Disjointization}
As introduced above, the overlaps among effective areas are the main reason for the unsatisfying query performance of R-tree. Therefore, eliminating overlaps among them is effective to enhance the query performance guarantee. However, this is infeasible for conventional R-trees since the stored rectangles are unchangeable. In contrast, the disjoint property of effective areas makes it achievable in our context.

In LSM-trees, two effective areas overlap when their associated range records invalidate the same key range. Since the latter range record carries more recent information, its efficacy dominates the earlier one within the overlapped key range. Thus, we can reorganize them to preserve only the most recent range record in each key interval to eliminate overlaps. This process we call {\it disjointization}.



\begin{figure*}[t]
\vspace{-6mm}
\centering
  \includegraphics[width=0.99\linewidth]{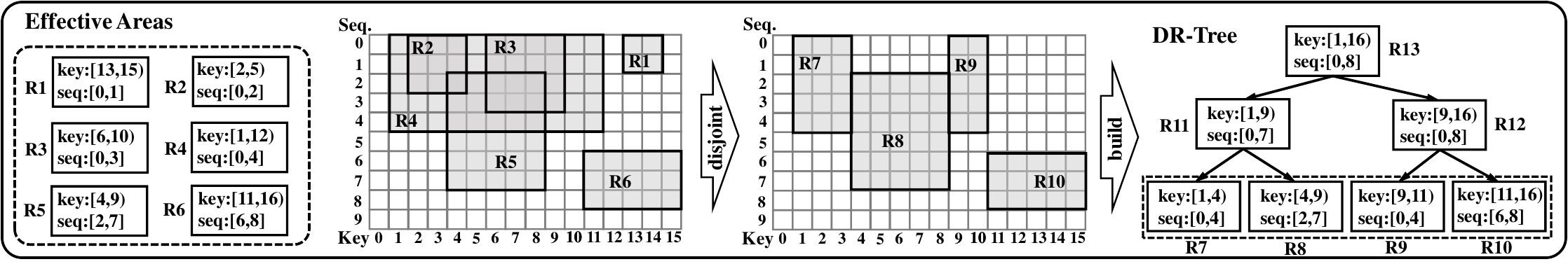} 
  \vspace{-4mm}
  \caption{An illustration of DR-tree construction from a set of effective areas.}
\label{fig:lsm_frtree_flush}
\vspace{-5mm}
\end{figure*}

Figure~\ref{fig:effective-area-merge} illustrates the three possible cases between two overlapping effective areas, denoted as $\alpha$ and $\beta$.
(a) When both the key and sequence ranges of $\beta$ are fully contained in those of $\alpha$, $\beta$ is completely dominated and can be directly replaced by $\alpha$.
(b) When $\beta$’s key range is contained in $\alpha$ but its sequence range only partially overlaps, $\beta$ updates part of $\alpha$’s key range while the remaining region is still dominated by $\alpha$. In this case, $\alpha$ is split by the key range of $\beta$ in the working space. Notably, when both start and end keys differ for $\alpha$ and $\beta$, a new effective area would be introduced after the disjointization.
(c) When both the key and sequence ranges partially overlap, as the figure presents, the overlapped region is dominated by $\beta$ due to its higher sequence number. Hence, $\alpha$ is trimmed to disjoint $\beta$ in the key domain. 
A potential concern is that trimming $\alpha$ might expose obsolete entries not covered by $\beta$. This does not occur because an effective area inherently defines the key and sequence ranges where its range delete is valid. When such an area is created, no matched entries exist in the LSM-tree that were issued before its starting sequence number. Therefore, trimming this area does not compromise the correctness of the resulting effective areas.


\begin{lemma}
For a set of effective areas $\mathcal{A} = \{\alpha_1, \alpha_2, \cdots , \alpha_n \}$, disjointization generates a new set $\mathcal{B} = \{\beta_1, \beta_2, \cdots , \beta_m \}$, where an entry can be covered by at most one ${\beta_i \in \mathcal{B}}$ in the working space.
\label{lemma: effective area merge}
\end{lemma}

The above discussion examines the overlap scenarios between two effective areas. When multiple effective areas overlap, their disjointization can always be decomposed into these fundamental cases as stated in Lemma~\ref{lemma: effective area merge}. Though new elements would be generated after disjointization, the total number is no more than twice of the original set. For instance, Figure~\ref{fig:lsm_frtree_flush} illustrates disjointization with six effective areas that form a sequence of disjoint effective areas sorted by key order. These disjoint areas can then be recursively merged to construct higher-level parent nodes, leading to a tree structure termed the {\bf Disjoint R-tree (DR-tree)}. 

\vspace{1mm}
\noindent
{\bf Remark.}
Owing to the disjoint property, each key is covered by at most one leaf node (Lemma~\ref{lemma: effective area merge}), and consequently, only a single internal node needs to be accessed per level during query processing. This property provides a strong theoretical bound on the worst-case query complexity and substantially improves lookup efficiency compared with traditional R-trees.  


\vspace{1mm}
\noindent {\bf Build DR-tree from Effective Areas.}
Obviously, deriving disjoint effective areas and ordering them by key is essential for DR-tree construction. While pairwise disjointization across multiple overlaps is computationally expensive, we provide a more efficient approach.

As shown in Figure~\ref{fig:lsm_frtree_flush}, the disjointization results occupy separate key intervals and are outlined by the most recent effective area covering that range. Hence, we can scan the effective areas along the key universe to determine the boundaries of these key intervals and record their dominant effective area.  
To this end, we use two min-heaps to store the start and end keys of all effective areas and scan them by repeatedly extracting their smallest key. Meanwhile, a max-heap, {\it curr}, is employed to maintain currently active areas, whose start keys have been processed but end keys remain pending. These effective areas are ordered by sequence numbers within this max-heap, ensuring that its top always tracks the dominant effective area for the preceding key interval.
 
A new interval begins during scanning whenever a start key has a larger sequence number than the current dominant area, as illustrated by R4, R5, and R6 in the figure. As a more recent element arrives, the top of {\it curr} is replaced and temporarily terminates its efficacy at this point, which generates a new record in the result, as R7 does. It is worth mentioning that the replaced effective area is still kept in {\it curr} since the end key has not been reached. Hence, it is possible to become the dominant effective area again, like R9.
Then, when the end key of the current dominant area is encountered, it indicates the termination of this area that should be recorded with the preceding key interval. Besides, if {\it curr} remains non-empty, the endpoint simultaneously marks the beginning of the next key interval (e.g., R9). By iterating through all start and end keys in this manner, we obtain a complete set of disjoint effective areas naturally ordered by key.

Whereas updating an RD-tree on disk is still costly since the inserted effective areas would change the layout of leaf nodes. To address this issue, we further design {\bf LSM-DRtree} that maintains the outstanding query efficiency of the RD-tree while achieving more favorable update performance.

\vspace{1mm}
\noindent {\bf LSM-DRtree Structure.}
The LSM-DRtree introduces an LSM-style design for the DR-tree to significantly enhance its update performance. A similar idea was also explored in~\cite{shin2021lsm} that introduces LSM-Rtree to enhance R-tree's update efficiency. However, our approach targets a fundamentally different structure, DR-tree, which requires distinct configurations, flushing method, and compaction mechanisms. LSM-DRtree maintains an in-memory R-tree as the write buffer and organizes multiple on-disk levels, each storing a DR-tree. Similar to the LSM-tree, the capacity of each level grows by a fixed size ratio. Under this design, once a DR-tree is written to disk, it remains immutable until the next compaction, thus effectively avoiding the costly in-place updates required when directly storing an entire RD-tree on disk.
Moreover, although updates to the write buffer are performed in memory and do not incur extra I/O cost, frequently updating a DR-tree would still impact the system performance.
To mitigate this, we employ an R-tree as the write buffer to further accelerate updates. 
We next detail the flush and compaction workflows for the index construction.

\vspace{1mm}
\noindent {\bf Construction of LSM-DRtree.}
Upon receiving a range record, its effective area is inserted into the in-memory R-tree. When the R-tree reaches its maximum capacity, a flush process is triggered. During this process, all effective areas stored in the R-tree’s leaf nodes are extracted to construct a DR-tree as introduced earlier. The resulting DR-tree is serialized level by level, ensuring that nodes within the same level are physically clustered and sorted by key range. 
Since flushes occur less frequently and the write buffer is limited in size, this process incurs moderate overhead.

When a disk level becomes full, a compaction process is triggered to merge its DR-tree with that of the next level. 
Benefiting from the fact that effective areas in each DR-tree are already disjoint and sorted by key order, the merging can be performed entirely through sequential I/O, avoiding random disk access. The compaction proceeds as a simple two-way merge: an iterator is built for each DR-tree, and at each step, the system retrieves the effective areas with the smallest start keys from both iterators, performs disjointization between them if they overlap, and proceeds iteratively until all effective areas have been processed. 

\vspace{1mm}
\noindent
{\bf Remark.} The compaction workflow of LSM-DRtree only involves pairwise disjointization between two effective areas at a time, avoiding the more intricate process for building DR-tree from arbitrary effective areas. Moreover, its streaming nature enables on-the-fly execution without introducing additional I/O overhead, thereby ensuring high update efficiency. Compared with LSM-Rtree, LSM-DRtree offers a much simpler merging process that avoids complex spatial alignment among rectangles, which achieves approximately 11\% faster construction latency in our experiments.
We next describe how LSM-DRtree supports range deletes and point lookups in LSM-based key-value stores and quantitatively analyze the operational overhead.



\noindent 
{\bf Range Delete with Global Index.}
Each range delete inserts a range record into the LSM-DRtree. The record is represented as a rectangular effective area that is stored by two ends of the key range and two sequence numbers. Meanwhile, the size of the sequence number is typically much smaller than the keys. Hence, the size of each range record can be formulated as $2k$. With the LSM structure, the range record would be rewritten during flush or compaction processes, which would incur subsequent overhead. 

\begin{lemma}
Updating a global index with $\tfrac{N}{\lambda}$ range records leads to $O(\frac{k}{B} \cdot T \cdot \log_T(\frac{N}{\lambda F}))$ amortized I/O cost.
\label{lemma: update lsm-rtree}
\end{lemma}

\begin{proof}[Proof Sketch]
The $\tfrac{N}{\lambda}$ range records produce at most $\tfrac{2N}{\lambda}$ effective areas after disjointization. To store them, the LSM-DRtree should introduce $L' = O(\log_{T}(\frac{N}{\lambda F'}))$ levels, where $F'$ and $T$ are the write buffer size and size ratio, respectively. Though the write buffer of the global range record index is typically smaller than that of the LSM-tree, $F'=O(F)$ still holds. Within an LSM-DRtree, the effective areas can eventually be stored in the bottommost level, after experiencing $O(TL')$ times compactions. Therefore, the overall construction I/O cost is $O(\frac{2N}{\lambda} \cdot \frac{T \cdot 2k}{B} \cdot \log_{T}(\frac{N}{\lambda F}))$, leading to the range delete cost for each operation as $O(\frac{k}{B} \cdot T \cdot \log_{T}(\frac{N}{\lambda F}))$ I/Os. 

\end{proof}

\vspace{-3mm} 
\noindent
{\bf Remark.}
By the above lemma, the update cost is lower than that incurred by the local range delete method, particularly for typical cases where $\lambda$ is large.

\vspace{1mm} 
\noindent {\bf Point Lookup with Global Index.}
The point lookup begins with searching the target key in an LSM-tree. If no matching key is found, the query terminates immediately without accessing the LSM-Rtree. As previously discussed, this leads to the I/O cost of $\lceil \varphi \log \! \tfrac{N}{F} \rceil$.  Otherwise, if a matching entry is found, the global index should be checked to verify whether it has been invalidated by subsequent range deletions. 

\begin{lemma}
After inserting $\tfrac{N}{\lambda}$ range records, global index checks validity of an entry with $O(\log_T^2(\frac{N}{\lambda F}))$ I/O cost in the worst case.
\label{lemma: query lsm-rtree}
\end{lemma}

\begin{proof}[Proof]
Following the design philosophy of LSM structures, checking the global index proceeds the LSM-DRtree level by level. In order to figure out the overall query cost, we first analyze the associated costs at the $i$-th level that contains $Q_i$ range records. As discussed, these range records result in up to $2Q_i$ effective areas after disjoinitzation that are stored in the leaf nodes of a DR-tree. Following Lemma ~\ref{lemma: effective area merge}, only one leaf node may cover the target key, which requires one I/O to retrieve. Moreover, storing these effective areas requires a DR-tree with $O(\log_D(Q_i))$ levels, and one node would be accessed in each level during the searching process due to the property of DR-tree. Moreover, all $L'$ levels of the LSM-DRtree would be searched. Note that $Q_i=F'\cdot (L')^{i}$, we immediately have the cumulative cost:

\vspace{-3mm}
{\small
\begin{equation}
    Z_0 = \sum_{i}^{L'} \left(\log_DQ_i + 1\right)=  O\left(\frac{L' \cdot \left(L' + 1\right)}{2} \cdot \log_{D}{T} + L' \cdot \log_D{F'}\right)
\end{equation}
}

Also, we have $D\ge T$ and $F'=O(F)$. The query cost can then be simplified to $O((L')^2)=O(\log_T^2(\frac{N}{\lambda F}))$. 

\end{proof}

\vspace{-2mm} 
\noindent
{\bf Remark.} After issuing $N/\lambda$ range deletes, the local range delete method incurs I/Os linear to $N/\lambda$ for point lookups. The above lemma shows that {\ourmethod} significantly improves this cost to poly-logarithmic scale regarding $N/\lambda$.

\vspace{1mm}

\begin{figure}[t]
\vspace{-2mm}
\centering
  \includegraphics[width=0.99\linewidth]{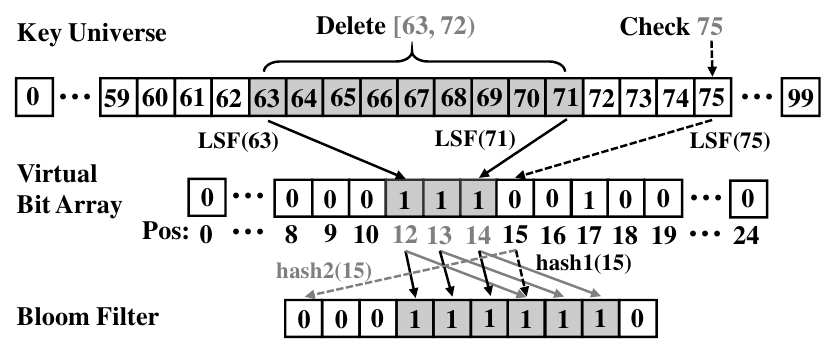} 
  \vspace{-3mm}
  \caption{The range-aware estimator (RAE) uses the virtual bit array to encode key ranges in a Bloom filter.}
\label{fig:filter_bloom_structure}
\vspace{-4mm}
\end{figure}

\subsection{Predictive Shortcut for Point Lookup}
\label{subsec:filter}
The global range record index can efficiently determine the validity of entries during point lookup. However, it is noticed that, according to the introduced point lookup workflow, all entries would be checked, including the valid ones that are not included in the global index. This observation motivates a predictive shortcut: if valid entries can be quickly identified before entering the global index with a certain probability, the results can be directly returned without incurring global index related overhead, thereby, further reducing the amortized point lookup cost.

This requires a dedicated component to store range record information and determine the coverage of individual keys without false negative estimation, ensuring that obsolete entries are not mistakenly returned during point lookups. Furthermore, this component is expected to be efficient and lightweight to avoid excessive overhead that otherwise degrade system performance, as formally defined in Problem ~\ref{prob:validity filter}. 
In addition, since range deletes are continuously applied to the database, this structure must deliver desirable dynamic performance. Hence, we propose {\bf entry validity estimator (EVE)}, a practical design that well supports above requirements with its efficacy evaluated in our experiments.


\begin{problem}
Given a set $S$ of $n$ key ranges within an integer universe $U = \{0, \ldots, U-1\}$, build a space-efficient in-memory data structure that answers emptiness queries of the form $q \cap S \neq \emptyset$? for any $q$ in $U$. The structure is allowed to return ``not empty'' when actually $q \cap S = \emptyset$ (i.e. false-positive error) with certain probability.
\label{prob:validity filter}
\end{problem}

A naive solution is to employ a Bloom filter to store every key within a range record. 
However, each range record corresponds to a large number of individual keys and inserting all of them into the Bloom filter leads to a significantly degraded false positive rate (FPR). As a result, many valid entries are likely to be incorrectly identified as obsolete, thus still requiring to consult global index. Moreover, processing every key within a range record further incurs significant computational costs, making this approach impractical.

Therefore, it is crucial to leverage range information inherent in range deletes. For this purpose, we propose a {\bf range-aware estimator (RAE)} that integrates a range encoding strategy with the Bloom filter, as presented in Figure~\ref{fig:filter_bloom_structure}. For RAE, a linear scaling function is employed to map the key universe into a bit array, where each bit represents a key range rather than an individual key. Consequently, a deleted key range can be represented by only a few bits. Then, the positions of these bits are recorded in the Bloom filter. In this process, the bit array used for range encoding is only a conceptual aid to derive bit positions that does not need to be physically stored. Hence, we refer to it as a \emph{virtual bit array}. 
This approach substantially reduces the number of insertions and greatly improves both estimation accuracy and efficiency. 

In Figure~\ref{fig:filter_bloom_structure}, when a range delete is issued, the two boundaries of its key range are mapped through the linear scaling function to identify the occupied virtual bit segment [12,14]. The positions within this segment are then inserted into the Bloom filter. For checking key 75 during point lookups, it is similarly mapped onto the virtual bit array to check the position in the Bloom filter. Since the Bloom filter returns a negative result, this key is definitely not covered by any range delete, and thus the corresponding entry can be directly returned without probing the global range record index.

\begin{figure}[t]
\vspace{-2mm}
\centering
  \includegraphics[width=0.99\linewidth]{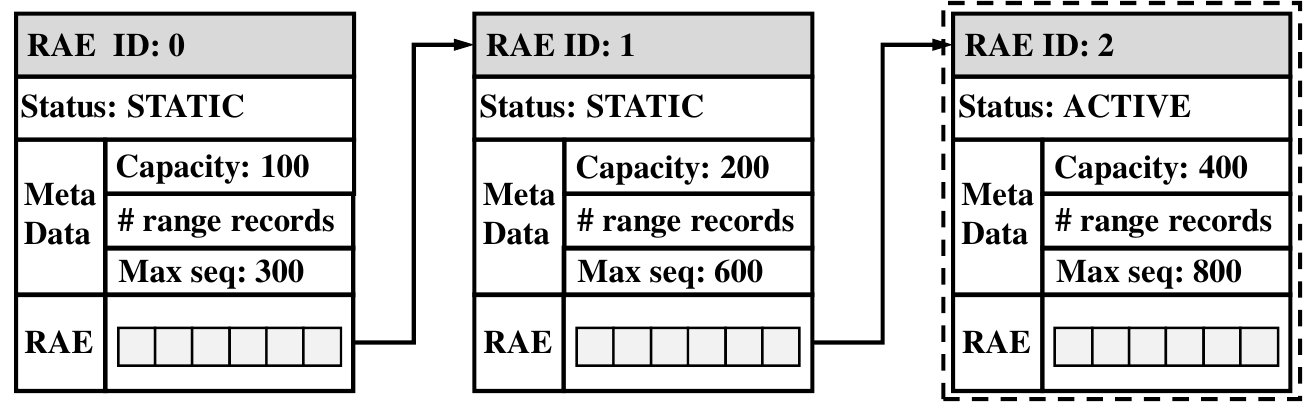} 
  \vspace{-3mm}
  \caption{The structure of the entry validity estimator (EVE).}
\label{fig:filter_hybrid_structure}
\vspace{-7mm}
\end{figure}

Meanwhile, with more range records inserted, the estimation accuracy of a single RAE can gradually degrades, and eventually becoming ineffective. This necessitates an extendable structure that can scale with the growing number of range deletes to support the application of key-value stores.
Therefore, we adopt a chained structure for the RAE, as illustrated in Figure~\ref{fig:filter_hybrid_structure}, which completes the design of our entry validity estimator.
To construct a EVE, an empty RAE is initialized before workload execution. This RAE accepts incoming range deletes and maintains the smallest and largest sequence numbers among them. When the current RAE reaches its capacity, a new active estimator with doubled capacity is created, while the previous one is marked as static. During a point lookup, once the target entry is found in the LSM-tree, the EVE is queried. The active RAE is checked first. If it returns a positive result, the entry may have been deleted by a range delete and must be verified against the global range record index. Otherwise, the query proceeds recursively to the most recent RAE until either a positive result is obtained or the entry’s sequence number exceeds that of the RAE. If all subfilters return negative, the entry is guaranteed to be valid and can be returned without further cost. 

\vspace{1mm}
\noindent
{\bf Discussion.} Therefore, if the target entry is valid, EVE can reduce the point lookup cost from $O(\log_{T}^2 (N /\lambda F))$ to  $O(\varepsilon \cdot (\log_{T}^2 (N /\lambda F)))$, where $\varepsilon$ is the false positive rate of EVE. 
We acknowledge that similar functionality is also achievable for certain dynamic range filters. However, these filters are typically designed to store individual keys hence struggle to offer sound performance in this case,
particularly under limited memory budgets. Moreover, their efficacy is confined to the key space, whereas EVE can leverage sequence number with its chained design. Nevertheless, we compare EVE with several competitive range filters in Section ~\ref{sec:evaluate}, which achieves the highest accuracy with over 20\% lower false positive rate.


\begin{figure*}[t]
\vspace{-6mm}
\centering
  \includegraphics[width=0.98\linewidth]{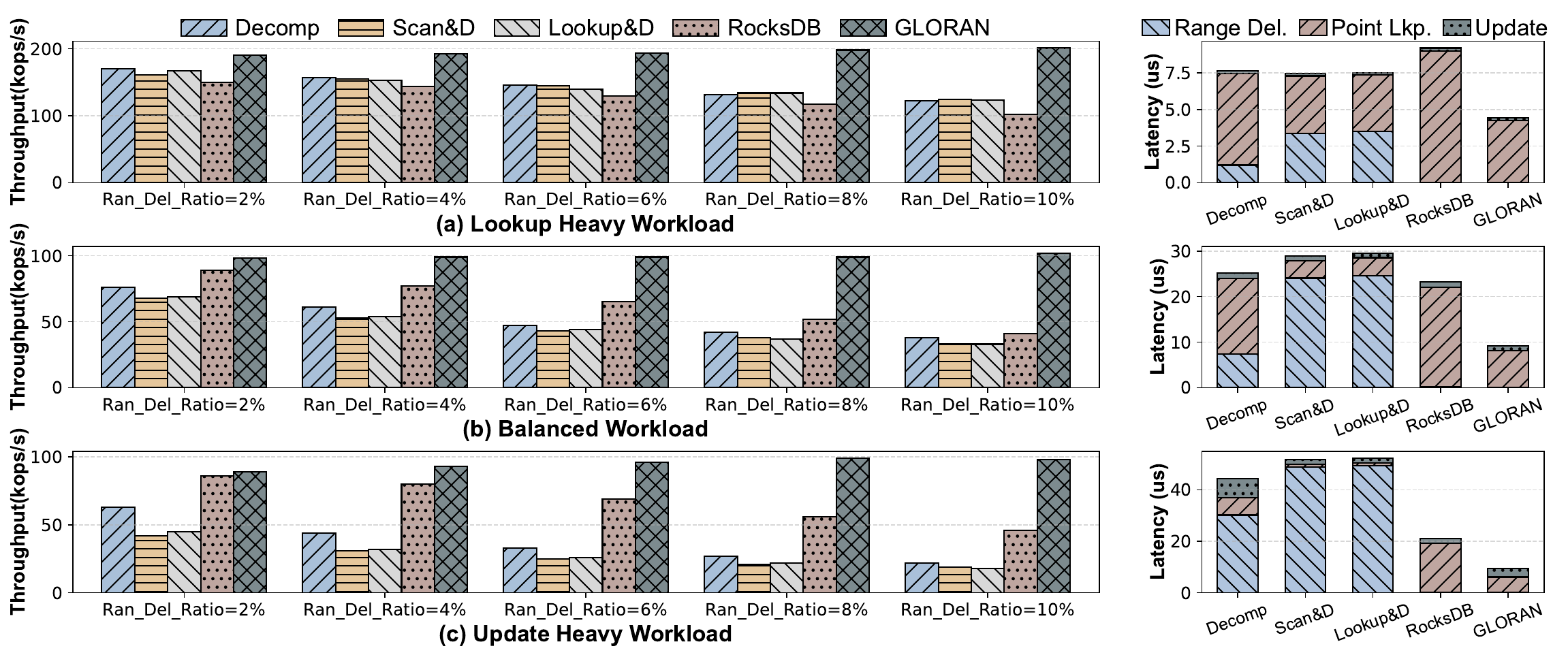} 
  \vspace{-3mm}
  \caption{We evaluate various range delete methods across diverse workloads. The results show that GLORAN consistently delivers the best performance across all test cases while enabling fast range deletes and efficient point lookups.}
\label{fig:group_1_workload}
\vspace{-3mm}
\end{figure*}

\subsection{Space Amplification of GLORAN}
The space amplification caused by obsolete entries under {\ourmethod} is comparable to that of other methods, since these entries can be properly reclaimed during compaction. In addition, range deletes introduce extra space overhead for storing deletion metadata. For this part, {\ourmethod} exhibits satisfying space efficiency, as each range record occupies a compact size, unlike the approaches that insert a tombstone for every key. We further compare the size of our global index with the range tombstone blocks used in LRR, another range record method. At the $i$-th level of the LSM-DRtree, which contains $Q_i$ range records, a DR-tree maintains at most $2Q_i$ leaf nodes after effective area disjointization. The number of nodes in higher levels decreases by a factor of $D$ from bottom to top. Denoting the total number of levels as $L_i$, the total number of nodes $R_i$ in this DR-tree can be expressed as:

\vspace{-2mm}
{\small
\begin{equation}
R_i = \frac{D}{D-1} \cdot \frac{D^{L_i}-1}{D^{L_i}} \cdot 2Q_i < \frac{D}{D-1} \cdot 2Q_i
\end{equation}
}

Aggregating the DR-trees across all levels of the LSM-DRtree, the overall size of the global range record index is bounded by $O(\frac{D}{D-1} \cdot Q \cdot k)$. Since $D$ is the DR-tree fanout (an integer greater than 2), $\frac{D}{D-1}$ is smaller than 2. Therefore, the space overhead of {\ourmethod} is $O(Q \cdot k)$, which is asymptotically the same as that of LRR. This result indicates that the space amplification of {\ourmethod} remains modest and well bounded.


In addition, {\ourmethod} incorporates a garbage collection (GC) strategy to clear obsolete range records from the global index to ensure space overhead and query efficiency. 
A range record becomes obsolete in two cases. The first occurs when its effective area is fully covered by a subsequent record, which is automatically handled by the LSM-DRtree compaction through effective area disjointization. The second arises when all entries it matches have been removed from the LSM-tree. To identify such cases, we attach an event listener to detect compactions on the bottommost level and record the largest sequence number of the resulting data. After such compaction, GC is triggered to purge elements whose sequence numbers are below this threshold.
Since outdated records are typically concentrated at the bottommost level of LSM-DRtree, GC is confined to this level to reduce overhead. Moreover, the same threshold is also used to drop outdated RAEs in the entry validity estimator, for further improving its accuracy and query efficiency.

\vspace{-2mm}
\section{Related Work}
\noindent\textbf{Key-value Stores Enhancement.} 
LSM-based key–value stores are applicable to a wide range of use cases across diverse application domains~\cite{mo2025aster, wang2023mirrorkv, chen2023chainkv, pang2021arkdb, yu2024lsmgraph, chen2021block}.
Research on LSM‑based key–value stores mostly targets enhancing the system performance, including overall throughput~\cite{luo2018toain} and space efficiency~\cite{lv2025rethinking} via improving different operations like point lookups, range queries, and updates. Point lookup improvements refine or replace Bloom filters~\cite{dayan2017monkey,zhang2018elasticbf,zhu2021reducing,dayan2021chucky,li2022seesaw}, adapt data distribution~\cite{zhang2022sa,zhang2022bi}, and exploit spatial locality~\cite{vu2021incremental,eldawy2021beast}. Range queries benefit from succinct tries, hierarchical or hybrid range filters, and learned or adversarial‑resistant models~\cite{zhang2018surf,luo2020rosetta,mossner2023bloomrf,knorr2022proteus,wang2023rencoder,vaidya2022snarf,oasis2024,costa2024grafite}. Update throughput boosts via compaction scheme variants~\cite{pan2017dcompaction, ahmad2015compaction, thonangi2017log, yao2017light, raju2017pebblesdb, idreos2019designcontinuum, mo2023learning, wu2015lsmtrie}, key–value separation~\cite{lu2017wisckey, chan2018hashkv}, learned indexes~\cite{dai2020wisckey}, and parallelism~\cite{huang2021nova, yu2022treeline}. While memory allocation tuning~\cite{luo2020breaking,kim2020robust} and finer‑grained compaction cut space amplification~\cite{dayanspooky,alkowaileet2019lsm,mao2020comprehensive}.

\vspace{1mm}
\noindent\textbf{Range Deletion Optimization.} There are several prior works~\cite{gartner2001efficient, lilja2006online, bhattacharjee2007efficient} researching on the bulk deletion in relational databases. However, this typically follows the direction for read enhancement, which pursues quick location and entry clustering through sorting or hashing, which is different from the range deletion in LSM-trees.
Nie {\it et al.} ~\cite{nie2024zone} tailored deletion for specified device like ZNS SSD. Sarkar {\it et al.} ~\cite{sarkar2020lethe} proposed Lethe to reduce the lifetime of tombstones with FADE, a novel compaction strategy. Whereas, these works mainly focus on the optimization for point deletion, which requires different techniques compared to our target operation range delete. Moreover, Lethe also discussed range deletes on a secondary delete key, where the range deletion is not applied to the primary key in the LSM-tree. This requires optimization concerning the data layout which falls within another field.

\vspace{-4mm}
\section{Evaluation}
\label{sec:evaluate}
This section presents evaluation results of different methods across diverse test cases. All experiments are conducted on a server equipped with a 13th Gen Intel(R) Core i9-13900K CPU @ 5.80GHz, 128GB DDR4 RAM, and a 2TB NVMe SSD, running 64-bit Ubuntu 22.04.5 LTS on an ext4 file system.

\vspace{1mm}
\noindent{\bf Implementation.}
We implement {\ourmethod} on top of RocksDB~\cite{rocksdb}, a widely adopted key-value store used in many prior studies~\cite{dostoevsky2018, dayan2019log, sarkar2020lethe, huynh2021endure, wang2024grf, liu2024structural}. Our implementation offers various tuning factors to support flexible customization of {\ourmethod} structure, such as the memory buffer and size ratio of the global LSM-DRtree, as well as the fanout of DRtrees. For the entry validity estimator, the capacity and memory of each RAE are also adjustable. 

\begin{figure*}[t]
\vspace{-6mm}
\centering
  \includegraphics[width=0.99\linewidth]{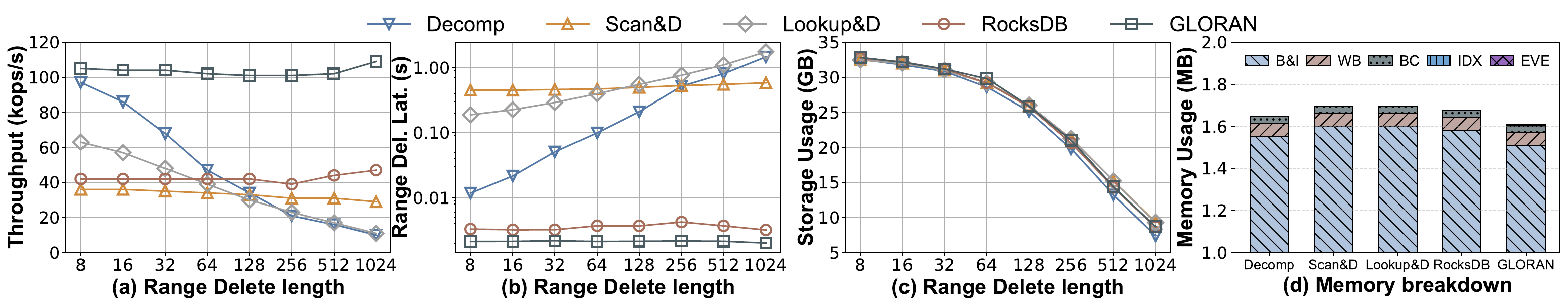} 
  \vspace{-3mm}
  \caption{The performance of different range delete methdos under various range delete lengths.}
\label{fig:group_2}
\vspace{-5mm}
\end{figure*}

\vspace{1mm}
\noindent{\bf Baselines and Implementation.}
We include all existing range delete methods for LSM-based key-value stores and build them on RocksDB~\cite{rocksdb}.
The decomposition method (abbr. {\decom}) deletes all keys within the target range individually using the {\it Delete} API. The lookup-and-delete (abbr. Lookup\&D) performs a point lookup for each key with {\it Get} API and only {\it Delete} the existing ones. The scan-and-delete method (abbr. Scan\&D) substitutes the point lookups with a range scan employing an {\it iterator} to identify the existing keys for deletion. The local range record method (abbr. RocksDB) is natively supported in RocksDB, so we use the built-in {\it DeleteRange} API directly. Moreover, Lethe~\cite{sarkar2020lethe} also discusses deletions in LSM-based key-value stores. Whereas it focuses on a fundamentally different operation, secondary range delete. Nevertheless, we also evaluate it in a case study based on the code provided by the authors~\cite{lethecode}.

\vspace{1mm}
\noindent{\bf Experimental Setting.}
In our experiments, we retain most parameters as RocksDB’s default settings, which are developer-optimized for general SSD-based applications~\cite{tuningguide}. In this setting, the memory buffer is configured to 64 MB. To enrich the evaluation, we also vary this value from 8 MB to 256 MB (Figure~\ref{fig:group_3_2} (a)). Besides, the Bloom filter is assigned 10 bits per entry. For the global range record index in {\ourmethod}, we explore different memory buffer sizes (Figure~\ref{fig:group_3_1} (b)) and size ratios (Figure~\ref{fig:group_3_1} (c)), with defaults of 4 MB and 10, respectively. For EVE, the first RAE is configured to hold 0.8 million range records with 10 bits per record, which we vary in Figure~\ref{fig:group_4} (c).
We evaluate these methods under workloads with different ratios of updates and point lookups. To access the impact of range deletes, we replace varying percentages of updates with range deletes (Figure~\ref{fig:group_1_workload}) and adjust their lengths to measure the impact on throughput, space amplification, and memory usage (Figure~\ref{fig:group_2}). Each test runs 100 million randomly generated operations on an empty database, with direct I/O and block cache enabled following prior work~\cite{sarkar2020lethe,dayan2019log,liu2024structural,mo2023learning,wang2024grf,huynh2024towards}. Each entry consists of a 256-byte key and an 768-byte value, which are varied in Figure~\ref{fig:group_3_1} (a) and (b), respectively. To ensure comprehensive evaluation, we also assess the influence of diverse block cache configurations (Figure~\ref{fig:group_3_1} (d)), data scales (Figure~\ref{fig:group_3_1} (c)), the efficacy of LSM-DR tree structure and EVE (~\ref{fig:group_4}), as well as the impact of range lookups (Table ~\ref{tab:range lookup}). In addition, we further include two practical benchmarks that are widely adopted in previous works~\cite{huynh2021endure,mo2023learning,mo2025grow}, RocksDB micro benchmark, db\_bench~\cite{dbbench} (Table ~\ref{tab:dbbench}), and YCSB~\cite{ycsb} (Table ~\ref{tab:ycsb}) with different key distributions to assess more aspects of these range delete methods.

\vspace{1mm}
\noindent{\bf Overall Comparison.}
We evaluate different range delete methods under three representative workloads: lookup heavy (90\% point lookups 10\% updates), balanced (50\% point lookups 50\% updates), and update heavy (10\% point lookups 90\% updates). To assess these methods with different amounts of range deletes, we vary the range delete ratio from 0\% to 10\%.
As Figure~\ref{fig:group_1_workload} presented, {\ourmethod} consistently delivers the highest throughput across diverse workloads. Notably, in the balanced workload, it outperforms the second-best method, {\rocksdb}, by up to 2.4×. In this workload, both lookup and range delete efficiency critically impact overall performance. This demonstrates {\ourmethod}'s ability to support efficient range deletes without sacrificing point lookup performance. By contrast, point-based methods incur significant overhead from range deletes, while {\rocksdb} suffers slower point lookups. Even in workloads, where baselines may partially mask their weaknesses, {\ourmethod} maintains superior performance, achieving up to $1.6\times$ higher throughput than {\scan} in lookup heavy workloads and $2.1\times$ higher than {\rocksdb} in update heavy workloads.

We further observe that baseline performance degrades noticeably as the range delete ratio increases. In the balanced workload, {\rocksdb}'s throughput drops by over 30\% when the range delete ratio increases from 2\% to 10\%, as more range tombstones are being checked during point lookups. Other methods also suffer from increased range deletes due to more tombstone insertions. In contrast, the efficient global index and effective EVE bring higher robustness for  {\ourmethod} towards varying range delete numbers.

These experiments indicate the outstanding range delete and point lookup performance of {\ourmethod} at a high level. We next provide a finer performance breakdown to validate this conclusion.

\vspace{1mm}
\noindent{\bf Exp 1: Point Lookup and Range Deletes.} 
We profile the performance of each method when the range delete ratio is 10\% and decompose the latency into point lookup, update, and range delete counterparts, as depicted in Figure~\ref{fig:group_1_workload}. Notably, the methods using point deletion suffer from significant range delete overhead due to the large number of tombstone insertions or costly additional validation processes, particularly in update heavy workloads. Moreover, these tombstones could also introduce higher update overheads as illustrated in Figures~\ref{fig:group_1_workload} (b).
{\rocksdb} exhibits much better range delete performance, while it still spends substantial time on lookups which is 3.4 times slower than {\ourmethod}. 
{\it In contrast, {\ourmethod} achieves outstanding range delete latency and competitive lookup performance simultaneously.}

\vspace{1mm}
\noindent{\bf Exp 2: Performance v.s. Range Length.} 
We vary the range delete length to evaluate the overall throughput, range delete latency, space amplification, and memory usage of different range delete methods under balanced workload. Figures~\ref{fig:group_2} (a) and (b) demonstrates that {\ourmethod} consistently delivers the most satisfying overall throughput and range delete latency across various cases, particularly for long ranges. Additionally, the point deletion based methods show explicit degradation as the range length grows, especially for {\decom} that necessitates obvious more tombstone insertions with higher overhead. In contrast, {\rocksdb} and {\ourmethod} exhibit more robust performance against the increase in range length due to their range records based strategies. Whereas, {\ourmethod} delivers better overall throughput than {\rocksdb} by up to $2.7\times$ when the range delete length is 1024. {\it Hence, {\ourmethod} maintains desirable and robust performance across varying range delete lengths.} 

Different range delete lengths lead to varied disk size which is reported in Figure~\ref{fig:group_2} (c). Generally, longer deletes remove more entries from the LSM-tree, resulting in reduced data volume. This applies to all methods. Meanwhile, their disk usage is comparable because all of them can effectively reclaim space occupied by obsolete entries via compaction. Although {\rocksdb} and {\ourmethod} maintain extra structures for range records, the associated space overhead is slight due to their compact size. 
\noindent{\it This indicates {\ourmethod} achieves sound space amplification and memory efficiency.}

In addition, Figure~\ref{fig:group_2} presents the memory breakdown with delete length of 128. Clearly, Bloom filters and indexes (B\&I) dominate memory usage across all methods. While LSM-tree buffer (WB) and block cache (BC) occupy smaller and fixed amounts for each method. In comparison, {\ourmethod} uses extra memory for global index write buffer (IDX) and entry validity estimator (EVE). Whereas their sizes are minimal compared to other parts.


\begin{figure*}[t]
 \vspace{-6mm}
 \centering
 \captionsetup{justification=raggedright, singlelinecheck=false}
  \includegraphics[width=\linewidth]{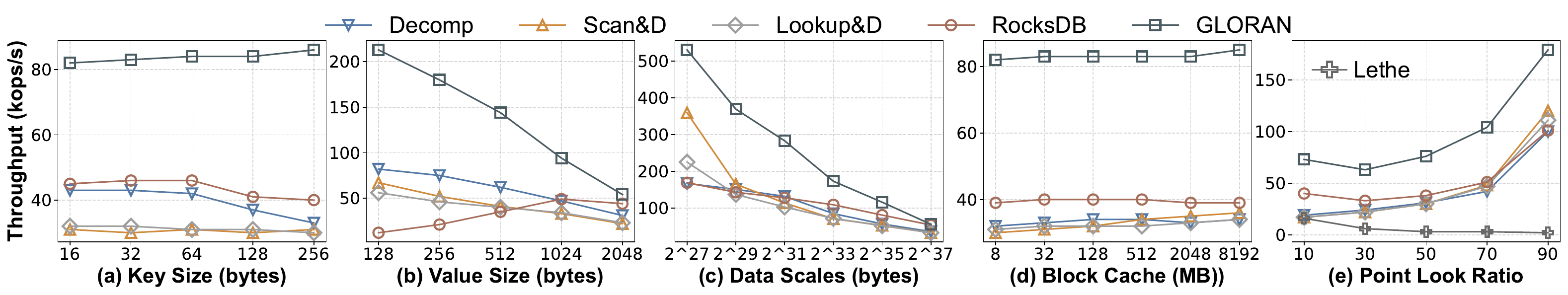} 
  \vspace{-7mm}
  \caption{Throughput across varying key sizes (a), value sizes (b), data scales (c), block cache sizes (d), and with Lethe (e).}
  \label{fig:group_3_1}
 \vspace{-3mm}
\end{figure*}

\begin{figure*}[t]
\vspace{-1mm}
  \centering
  \begin{minipage}[c]{0.37\textwidth}
  \captionsetup{singlelinecheck=true, font=small, justification=raggedright, margin={-2mm, -0mm}}
  \vspace{-4mm}
  \captionof{table}{Normalize throughput with range lookup.}
    \centering
    \vspace{-4mm}
\resizebox{0.99\linewidth}{!}{%
\begin{tabular}{|l|lllll|}
\hline
\multirow{2}{*}{} & \multicolumn{5}{c|}{Range Lookup Ratio (\%)} \\ \cline{2-6} 
    & \multicolumn{1}{c}{2} 
    & \multicolumn{1}{c}{4} 
    & \multicolumn{1}{c}{6} 
    & \multicolumn{1}{c}{8} 
    & \multicolumn{1}{c|}{10} \\ \hline
Decomp            
& $1.00\times$     
& $1.00\times$  
& $1.00\times$      
& $1.00\times$      
& $1.00\times$                           \\ \hline

Scan\&D           
& $1.02\times$     
& $1.11\times$     
& $1.03\times$     
& $1.04\times$     
& $1.01\times$                          \\ \hline

Look\&D         
& $1.05\times$     
& $1.13\times$     
& $1.00\times$      
& $1.06\times$     
& $1.00\times$                           \\ \hline

RocksDB 
& $1.35\times$     
& $1.37\times$     
& $1.23\times$    
& $1.23\times$     
& $1.17\times$                          \\ \hline

\textbf{GLORAN}   
& {\ul \textbf{$2.20\times$}} 
& {\ul \textbf{$2.04\times$}} 
& {\ul \textbf{$1.72\times$}} 
& {\ul \textbf{$1.62\times$}} 
& {\ul \textbf{$1.45\times$}}           \\ \hline
\end{tabular}
\label{tab:range lookup}
}
\end{minipage}
\hfill
\begin{minipage}[c]{0.62\textwidth}
    \centering
    \includegraphics[width=\textwidth]{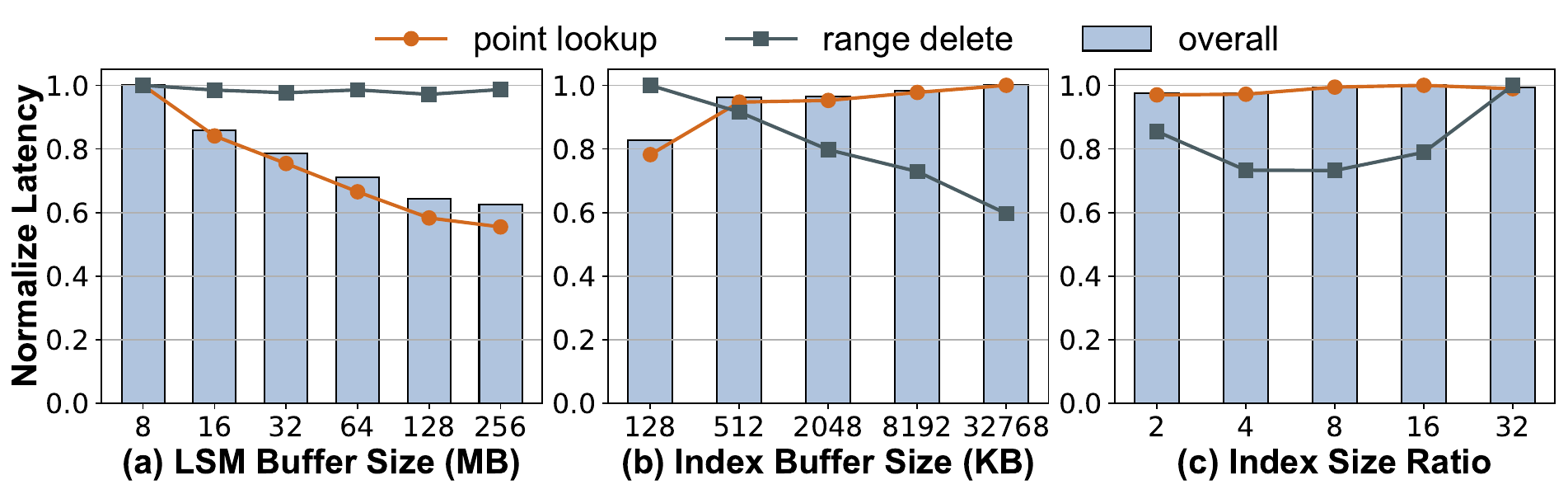}
  \end{minipage}
  \vspace{-4mm}
  \caption{The impact of range lookups (Table ~\ref{tab:range lookup}), LSM buffer size (a), as well as global index buffer size (b) and size ratio (c).}
  \label{fig:group_3_2}
\vspace{-5mm}
\end{figure*}

\vspace{1mm}
\noindent{\bf Exp 3: Performance v.s. Entry Sizes.}
We vary the key size and value size under balanced workload and show the results in Figure~\ref{fig:group_3_1} (a) and (b), respectively. In Figure~\ref{fig:group_3_1} (a), the key size is changed 
with entry size fixed at 1024 bytes. With larger key sizes is applied, the throughput of {\decom} decreases since larger tombstones are inserted. The performance of {\rocksdb} also drops for the more costly point lookups resulted by enlarged range tombstones. 
By contrast, the LSM-DRtree and EVE facilitate {\ourmethod} with efficient point lookups hence achieves stable performance that is up to $2.2\times$ improvement over {\rocksdb}.
Figure~\ref{fig:group_3_1} (b) varies value size
with key size fixed to 64 bytes. For all the methods, their throughput decreases as value size increases due to the increased compaction frequency and the associated overhead. Moreover, the enlarged value incurs deeper LSM-tree that requires higher point lookup cost. This eventually dominates the overall overhead and progressively mitigates the performance difference among diverse methods. Nevertheless, {\ourmethod} still maintains 20\% higher throughput than the best baseline at 2048 bytes.
Interestingly, the performance of {\rocksdb} is noticeably lower than others with small value size. In such cases, the point lookups cost on LSM-tree is relatively low, thus amplifying the impact incurred by checking range tombstones. As the results show, {\ourmethod} achieves $10.6\times$ faster point lookups than {\rocksdb} when value size is 256 bytes. This also explains why {\rocksdb} is outperformed by other baselines, though they delivers considerably less overall throughput than {\ourmethod}.

\vspace{1mm}
\noindent{\bf Exp 4: Performance v.s. Data Scales.}
We incorporate more operations in the balanced workloads to test the range delete methods under diverse data scales, which varies from $2^{27}$ bytes to $2^{37}$ bytes in Figure~\ref{fig:group_3_1} (c). With increased data volume, the LSM-tree develops more levels that brings higher update and lookup costs. Hence generally impact the throughput of all methods. Meanwhile, this incurs more overhead for tomsbtons insertion that impact the point deleteion based methods.  
The increased operation number also issues more range records. Hence the point lookup performance of {\rocksdb} is drastically affected according to Table~\ref{tab:costs}. While {\ourmethod} presents obviously improved point lookup complexity that makes the deterioration resulted by scalded data volume more moderate. As a result, though the increases data scale poses challenge for {\ourmethod}, it still presents impressive performance and adaptability.

\vspace{1mm}
\noindent{\bf Exp 5: Performance v.s. LSM Parameters.}
We vary the block cache size from 8 MB to 8192 MB to show the normalized throughput of different methods in Figure~\ref{fig:group_3_1}(d). The results indicate that enlarging the block cache typically cannot obviously enhance the performance since it does not impact the workflow of range delete relevant operations. Although {\lookup} and {\scan} gain minor improvements from faster lookups, the benefits are limited, and {\ourmethod} still achieves superior performance.
Figure~\ref{fig:group_3_2}(a) reports {\ourmethod}’s performance under different LSM memory buffer sizes. For clarity, we present both normalized overall latency and the breakdown for point lookups and range deletes. In this figure, larger buffers mainly improve point lookup efficiency by reducing the number of LSM levels, thus lowering overall latency. Meanwhile, the overheads of other operations remain relatively stable.

\vspace{1mm}
\noindent{\bf Exp 6: Performance v.s. Global Index Parameters.}
We evaluate various configurations of the global range record index by tuning the LSM-DRtree parameters. As shown in Figures~\ref{fig:group_3_2}(b), larger memory buffers reduce range delete overhead by lowering the LSM-DRtree depth and compaction frequency. However, they also increase the DR-tree depth at each level, resulting in higher point lookup costs. This contrasts with the LSM-tree in Figures~\ref{fig:group_3_2}(a), where sorted entries and fence pointers make lookup cost less sensitive to level size. Similarly, the range deletes also benefit from an increased size ratio in Figure~\ref{fig:group_3_2}(c). While the performance degrades when the ratio exceeds 8 due to the high cost of merging enlarged levels. These findings highlight the need to carefully tune LSM-DRtree parameters for different workloads and suggest opportunities for further optimization.

\begin{figure*}[t]
\vspace{-6mm}
  \centering
  \begin{minipage}[t]{0.65\textwidth}
    \vspace{0pt} 
    \centering
    \includegraphics[width=0.99\textwidth]{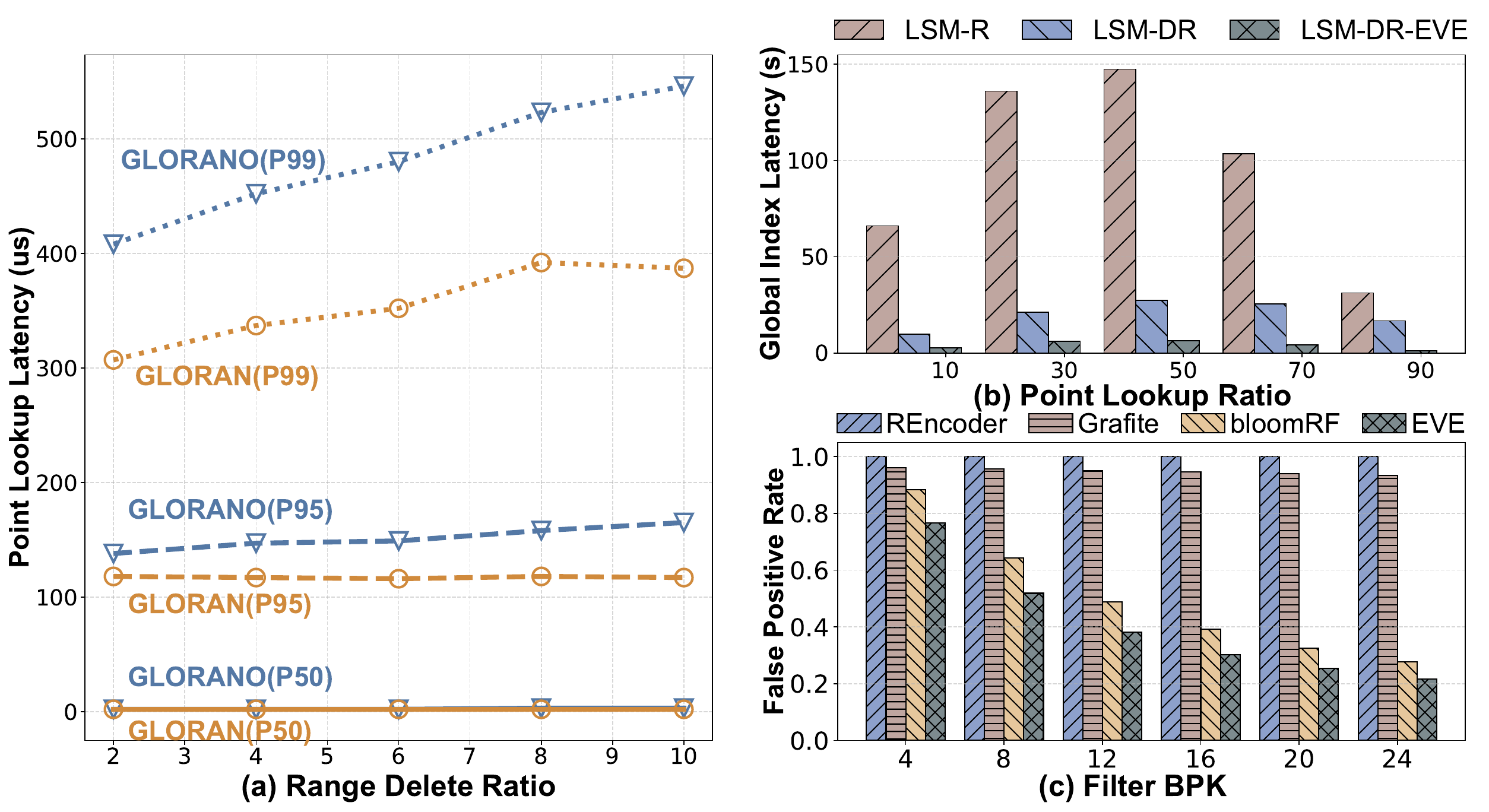}
  \end{minipage}
  \hfill
  \begin{minipage}[t]{0.34\textwidth}
    \vspace{0pt} 
    \captionsetup{singlelinecheck=false, font=small, justification=raggedright, margin={0mm, -0mm}}
    \captionof{table}{Normalize throughput under YCSB.}
    \centering
    \vspace{-4mm}
    \resizebox{0.98\linewidth}{!}{%
    \begin{tabular}{|l|llll|}
      \hline
      \multirow{2}{*}{} & \multicolumn{4}{c|}{Workload} \\ \cline{2-5}
      & Point-L & Balance & Update & Range-L \\
      \hline
      Decomp 
      & $1.01\times$ 
      & $1.00\times$ 
      & $1.29\times$ 
      & $1.00\times$ \\ \hline
      Scan\&D 
      & $1.02\times$ 
      & $1.16\times$ 
      & $1.03\times$ 
      & $1.17\times$ \\ \hline
      Look\&D 
      & $1.04\times$ 
      & $1.10\times$ 
      & $1.00\times$ 
      & $1.08\times$ \\ \hline
      RocksDB 
      & $1.00\times$ 
      & $1.11\times$ 
      & $1.63\times$ 
      & $1.08\times$ \\ \hline
      \textbf{GLORAN} 
      & {\ul \textbf{$1.29\times$}} 
      & {\ul \textbf{$2.11\times$}} 
      & {\ul \textbf{$3.87\times$}} 
      & {\ul \textbf{$1.54\times$}} \\ \hline
    \end{tabular}
    \label{tab:ycsb}
    }

    \vspace{3mm}
    \captionof{table}{Normalize throughput under db\_bench.}
    \vspace{-4mm}
    \centering
    \resizebox{0.98\linewidth}{!}{%
    \begin{tabular}{|l|lllll|}
      \hline
      \multirow{2}{*}{} & \multicolumn{5}{c|}{Point Lookup Ratio (\%)} \\ \cline{2-6}
      & 10 & 30 & 50 & 70 & 90 \\ \hline
      Decomp 
      & $1.04\times$ 
      & $1.06\times$ 
      & $1.04\times$ 
      & $1.00\times$ 
      & $1.12\times$ \\ \hline
      Scan\&D 
      & $1.01\times$ 
      & $1.00\times$ 
      & $1.00\times$ 
      & $1.14\times$ 
      & $1.22\times$ \\ \hline
      Look\&D 
      & $1.00\times$ 
      & $1.02\times$ 
      & $1.00\times$ 
      & $1.16\times$ 
      & $1.26\times$ \\ \hline
      RocksDB 
      & $3.81\times$ 
      & $2.56\times$ 
      & $1.99\times$ 
      & $1.66\times$ 
      & $1.00\times$ \\ \hline
      \textbf{GLORAN} & {\ul \textbf{$4.51\times$}} & {\ul \textbf{$3.08\times$}} & {\ul \textbf{$2.63\times$}} & {\ul \textbf{$2.54\times$}} & {\ul \textbf{$1.92\times$}} \\ \hline
    \end{tabular}
    \label{tab:dbbench}
    }
  \end{minipage}
  
  \vspace{-4mm}
  \caption{Results on additional benchmarks YCSB (Table~\ref{tab:ycsb}) and db\_bench (Table~\ref{tab:dbbench}). We further present point lookup latency (a) and global index query latency (b) to evaluate LSM-DRtree's efficacy, as well as compare EVE with other range filters (c).}
  \vspace{-4mm}
  \label{fig:group_4}
\end{figure*}


\vspace{1mm}
\noindent{\bf Exp 7: Range Deletes and Range Lookups.}
The range lookup operation retrieves all entries within a specified key range. To accomplish this, the system constructs iterators across all LSM levels and sequentially retrieves the matching keys. In {\rocksdb}, iterators should also be created for range tombstone blocks to determine key validity. Similarly, in {\ourmethod}, iterators are built for each DR-tree in the global index, as the effective areas are sorted and sequentially stored on disk.
In our experiments, we replace a portion of point lookups with randomly generated range lookups of length 100 and vary the proportion (range lookup ratio) from 0\% to 10\%. For ease of comparison, we normalize the throughput of all methods and report the results in Table~\ref{tab:range lookup}. Clearly, {\ourmethod} consistently outperforms all baselines, achieving more than 45\% higher throughput than the basic baseline. This demonstrates its sound capability in handling range lookup operations efficiently.


\vspace{1mm}
\noindent{\bf Exp 8: Performance of Lethe.}
Lethe~\cite{sarkar2020lethe} also discusses range deletion in LSM-tree. 
Whereas it targets other operation, {\it secondary range delete}, which is applied to secondary key. This is fundamentally different from our target. 
Nevertheless, we also evaluate its performance and compare with other methods. To this end, we adopt the code released by author~\cite{lethecode} and implement range deletes via the {\it deleterange} API. 
As Figure~\ref{fig:group_3_1} (e) illustrates, Lethe delivers relatively limited throughput that is consistently deteriorated as more point lookups are applied. This is not surprising because of the different design targets and techniques. To accelerate deletion on the secondary key, Lethe introduces KiWi, a new data layout tailored for its application.
However, this simultaneously complicates the workflow of other operations like point lookup, especially when range deletes are considered.

\vspace{1mm}
\noindent{\bf Exp 9: Micro Benchmark db\_bench.}
We further evaluate the range delete methods using the db\_bench microbenchmark provided in RocksDB~\cite{rocksdb} to generate test cases with different ratios of updates and point lookups.
We randomly replace 10\% of updates with range deletes, while varying the percentage of point lookups from 10\% to 90\%, and show the results in Table ~\ref{tab:dbbench}. Across all cases, {\ourmethod} achieves the best overall performance, reaching up to 1.5 times the throughput of the strongest baseline. Similar to Figure~\ref{fig:group_1_workload}, the range record-based methods like {\rocksdb} and {\ourmethod} are more competitive in update-heavy workloads. Nevertheless, {\ourmethod} still achieves 52\% higher throughput than the most performant point delete based method, {\lookup}, under lookup-heavy workload that only contains 10\% updates, indicating {\ourmethod}'s impressive performance of both range deletes and point lookups.

\vspace{1mm}
\noindent{\bf Exp 10: YCSB Benchmark and Zipfian Distribution.}
We introduce YCSB~\cite{ycsb}, another widely used benchmark to provide a more comprehensive assessment. We utilize four typical workloads under Zipfian distribution, {\it Point-L} contains 90\% point lookups and 10\% updates, {\it Balance} contains 50\% point lookups and 50\% updates, {\it Update} contains 10\% point lookups and 90\% updates, {\it Range-L} contains 20\% range lookups and 80\% updates. Since YCSB does not include range deletes, we randomly replace 10\% of updates with range deletes. According to Table~\ref{tab:ycsb}, {\ourmethod} delivers the most competitive performance across all four workloads. In particular, it achieves $3.9×\times$ improvements over the baseline in update-heavy workloads. It is observed that the enhancement in lookup-heavy is less significant. In Zipfian distribution, a certain part of keys are updated and queried more frequently. In this case, most queried keys are valid, thus mitigating the efficacy of EVE. Even though {\ourmethod} still apparently higher throughput than baseline by 1.3 and 1.5 times under  Point-L and Range-L workloads, respectively. These results confirm {\ourmethod}'s capability of handling various benchmarks and diverse access patterns.

\vspace{1mm}
\noindent{\bf Exp 11: Efficacy of LSM-DRtree and Entry Validity Estimator.}
Figure~\ref{fig:group_4} (a) compares the point lookup performance of {\ourmethod} with GLORANO, which utilizes LSM-Rtree as the global range record index. We prepare the database with 100M updates and replace different ratios of operations with range deletes. Then, we conduct 50M point lookups to collect the latencies and illustrate the median (P50) and tail (P95 and P99) values in the figure. Obviously, {\ourmethod} achieves lower tail latency under both percentiles. Specifically, its P99 latency is smaller than GLORANO by almost $1.5\times$ when 10M range deletes are included in the workload. Moreover, the performance of {\ourmethod} is more constant as the range delete ratio increases, while GLORANO presents more apparent rises in latency. Besides, both methods achieve a satisfying median performance. This indicates that {\ourmethod} can not only effectively mitigate point lookup overhead but also deliver more predictable and robust performance under diverse workloads. This owes to the two core designs of {\ourmethod}, the LSM-DRtree-based global index and range encoding filter, which are further assessed. 

Figure~\ref{fig:group_4} (b) exhibits the latency for querying the global index under three designs, global index using LSM-R and LSM-DR, and LSM-DRtree integrated with EVE (LSM-DR-REF). We include 100M operations to vary the point lookup ratio while keeping range delete ratio as 10\%. The results show that LSM-DR significantly outperforms LSM-R, which delivers up to $6.2\times$ quicker latency, by effectively reducing redundant range record access. With the integration of EVE, it further reduces global index latency by more than $3.5\times$ across all cases. This effectively evidences the efficacy the LSM-DRtree and entry validity estimator.

We further compare EVE with three competitive range filters that support dynamic updates: Grafite~\cite{costa2024grafite}, REncoder~\cite{wang2023rencoder}, and bloomRF~\cite{mossner2023bloomrf}. Figure~\ref{fig:group_4} (c) presents the their FPR when assigning diverse bits per range delete(BPK). Specifically, we insert 14M random ranges with a length of 100, followed by 1M random queries. The capacity of the first RAE is set to 2M for EVE. 
The results indicate that EVE consistently achieves the lowest false positive rate across all BPKs, especially under higher BPKs. This is under expectation since our estimator actually addresses the different problem of range filters, which highlights the capability of key range encoding instead of individual key points. Meanwhile, BloomRF also considers key ranges and coarsely encodes them with prefixes, thus performing better than other filters. However, our EVE still delivers more than 20\% better FPR across all memory budgets.

\vspace{-3mm}
\section{Conclusion}
This paper proposes a novel range delete method, {\ourmethod}, that efficiently manages the range records with a global index, which offers competitive and robust range delete performance as well as satisfying lookup performance. We further integrate it with an entry validity estimator to enhance it for better performance.

\bibliographystyle{ACM-Reference-Format}
\bibliography{arxiv}


\begin{thebibliography}{68}


\ifx \showCODEN    \undefined \def \showCODEN     #1{\unskip}     \fi
\ifx \showDOI      \undefined \def \showDOI       #1{#1}\fi
\ifx \showISBNx    \undefined \def \showISBNx     #1{\unskip}     \fi
\ifx \showISBNxiii \undefined \def \showISBNxiii  #1{\unskip}     \fi
\ifx \showISSN     \undefined \def \showISSN      #1{\unskip}     \fi
\ifx \showLCCN     \undefined \def \showLCCN      #1{\unskip}     \fi
\ifx \shownote     \undefined \def \shownote      #1{#1}          \fi
\ifx \showarticletitle \undefined \def \showarticletitle #1{#1}   \fi
\ifx \showURL      \undefined \def \showURL       {\relax}        \fi
\providecommand\bibfield[2]{#2}
\providecommand\bibinfo[2]{#2}
\providecommand\natexlab[1]{#1}
\providecommand\showeprint[2][]{arXiv:#2}

\bibitem[\protect\citeauthoryear{??}{tec}{2025}]%
        {technicalreport}
 \bibinfo{year}{2025}\natexlab{}.
\newblock \bibinfo{title}{Technical Report for GLORAN}.
\newblock \bibinfo{howpublished}{\url{https://anonymous.4open.science/r/GLORAN-56FF/GLORAN_Technial_Report.pdf}}.
\newblock


\bibitem[\protect\citeauthoryear{Ahmad and Kemme}{Ahmad and Kemme}{2015}]%
        {ahmad2015compaction}
\bibfield{author}{\bibinfo{person}{Muhammad~Yousuf Ahmad} {and} \bibinfo{person}{Bettina Kemme}.} \bibinfo{year}{2015}\natexlab{}.
\newblock \showarticletitle{Compaction management in distributed key-value datastores}.
\newblock \bibinfo{journal}{\emph{Proceedings of the VLDB Endowment}} \bibinfo{volume}{8}, \bibinfo{number}{8} (\bibinfo{year}{2015}), \bibinfo{pages}{850--861}.
\newblock


\bibitem[\protect\citeauthoryear{Alkowaileet, Alsubaiee, and Carey}{Alkowaileet et~al\mbox{.}}{2019}]%
        {alkowaileet2019lsm}
\bibfield{author}{\bibinfo{person}{Wail~Y Alkowaileet}, \bibinfo{person}{Sattam Alsubaiee}, {and} \bibinfo{person}{Michael~J Carey}.} \bibinfo{year}{2019}\natexlab{}.
\newblock \showarticletitle{An LSM-based Tuple Compaction Framework for Apache AsterixDB (Extended Version)}.
\newblock \bibinfo{journal}{\emph{arXiv preprint arXiv:1910.08185}} (\bibinfo{year}{2019}).
\newblock


\bibitem[\protect\citeauthoryear{Alsubaiee, Altowim, Altwaijry, Behm, Borkar, Bu, Carey, Cetindil, Cheelangi, Faraaz, et~al\mbox{.}}{Alsubaiee et~al\mbox{.}}{2014}]%
        {alsubaiee2014asterixdb}
\bibfield{author}{\bibinfo{person}{Sattam Alsubaiee}, \bibinfo{person}{Yasser Altowim}, \bibinfo{person}{Hotham Altwaijry}, \bibinfo{person}{Alexander Behm}, \bibinfo{person}{Vinayak Borkar}, \bibinfo{person}{Yingyi Bu}, \bibinfo{person}{Michael Carey}, \bibinfo{person}{Inci Cetindil}, \bibinfo{person}{Madhusudan Cheelangi}, \bibinfo{person}{Khurram Faraaz}, {et~al\mbox{.}}} \bibinfo{year}{2014}\natexlab{}.
\newblock \showarticletitle{AsterixDB: A scalable, open source BDMS}.
\newblock \bibinfo{journal}{\emph{arXiv preprint arXiv:1407.0454}} (\bibinfo{year}{2014}).
\newblock


\bibitem[\protect\citeauthoryear{Bhattacharjee, Malkemus, Lau, Mckeough, Kirton, Von~Boeschoten, and Kennedy}{Bhattacharjee et~al\mbox{.}}{2007}]%
        {bhattacharjee2007efficient}
\bibfield{author}{\bibinfo{person}{Bishwaranjan Bhattacharjee}, \bibinfo{person}{Timothy Malkemus}, \bibinfo{person}{Sherman Lau}, \bibinfo{person}{Sean Mckeough}, \bibinfo{person}{Jo-Anne Kirton}, \bibinfo{person}{Robin Von~Boeschoten}, {and} \bibinfo{person}{John Kennedy}.} \bibinfo{year}{2007}\natexlab{}.
\newblock \showarticletitle{Efficient Bulk Deletes for Multi Dimensionally Clustered Tables in DB2.}. In \bibinfo{booktitle}{\emph{VLDB}}. \bibinfo{pages}{1197--1206}.
\newblock


\bibitem[\protect\citeauthoryear{Chan, Liang, Li, He, Lee, Zhu, Dong, Xu, Xu, Jiang, et~al\mbox{.}}{Chan et~al\mbox{.}}{2018}]%
        {chan2018hashkv}
\bibfield{author}{\bibinfo{person}{Helen~HW Chan}, \bibinfo{person}{Chieh-Jan~Mike Liang}, \bibinfo{person}{Yongkun Li}, \bibinfo{person}{Wenjia He}, \bibinfo{person}{Patrick~PC Lee}, \bibinfo{person}{Lianjie Zhu}, \bibinfo{person}{Yaozu Dong}, \bibinfo{person}{Yinlong Xu}, \bibinfo{person}{Yu Xu}, \bibinfo{person}{Jin Jiang}, {et~al\mbox{.}}} \bibinfo{year}{2018}\natexlab{}.
\newblock \showarticletitle{HashKV: Enabling Efficient Updates in KV Storage via Hashing}. In \bibinfo{booktitle}{\emph{2018 USENIX Annual Technical Conference (USENIX ATC 18)}}. \bibinfo{pages}{1007--1019}.
\newblock


\bibitem[\protect\citeauthoryear{Chen, He, Li, and Luo}{Chen et~al\mbox{.}}{2024}]%
        {oasis2024}
\bibfield{author}{\bibinfo{person}{Guanduo Chen}, \bibinfo{person}{Zhenying He}, \bibinfo{person}{Meng Li}, {and} \bibinfo{person}{Siqiang Luo}.} \bibinfo{year}{2024}\natexlab{}.
\newblock \showarticletitle{Oasis: An Optimal Disjoint Segmented Learned Range Filter}.
\newblock \bibinfo{journal}{\emph{Proceedings of the VLDB Endowment}} (\bibinfo{year}{2024}).
\newblock


\bibitem[\protect\citeauthoryear{Chen, Li, Cai, Jia, Ju, Shao, and Shen}{Chen et~al\mbox{.}}{2023}]%
        {chen2023chainkv}
\bibfield{author}{\bibinfo{person}{Zehao Chen}, \bibinfo{person}{Bingzhe Li}, \bibinfo{person}{Xiaojun Cai}, \bibinfo{person}{Zhiping Jia}, \bibinfo{person}{Lei Ju}, \bibinfo{person}{Zili Shao}, {and} \bibinfo{person}{Zhaoyan Shen}.} \bibinfo{year}{2023}\natexlab{}.
\newblock \showarticletitle{ChainKV: A Semantics-Aware Key-Value Store for Ethereum System}.
\newblock \bibinfo{journal}{\emph{Proceedings of the ACM on Management of Data}} \bibinfo{volume}{1}, \bibinfo{number}{4} (\bibinfo{year}{2023}), \bibinfo{pages}{1--23}.
\newblock


\bibitem[\protect\citeauthoryear{Chen, Li, Cai, Jia, Shen, Wang, and Shao}{Chen et~al\mbox{.}}{2021}]%
        {chen2021block}
\bibfield{author}{\bibinfo{person}{Zehao Chen}, \bibinfo{person}{Bingzhe Li}, \bibinfo{person}{Xiaojun Cai}, \bibinfo{person}{Zhiping Jia}, \bibinfo{person}{Zhaoyan Shen}, \bibinfo{person}{Yi Wang}, {and} \bibinfo{person}{Zili Shao}.} \bibinfo{year}{2021}\natexlab{}.
\newblock \showarticletitle{Block-lsm: An ether-aware block-ordered lsm-tree based key-value storage engine}. In \bibinfo{booktitle}{\emph{2021 IEEE 39th International Conference on Computer Design (ICCD)}}. IEEE, \bibinfo{pages}{25--32}.
\newblock


\bibitem[\protect\citeauthoryear{Costa, Ferragina, and Vinciguerra}{Costa et~al\mbox{.}}{2024}]%
        {costa2024grafite}
\bibfield{author}{\bibinfo{person}{Marco Costa}, \bibinfo{person}{Paolo Ferragina}, {and} \bibinfo{person}{Giorgio Vinciguerra}.} \bibinfo{year}{2024}\natexlab{}.
\newblock \showarticletitle{Grafite: Taming Adversarial Queries with Optimal Range Filters}.
\newblock \bibinfo{journal}{\emph{Proceedings of the ACM on Management of Data}} \bibinfo{volume}{2}, \bibinfo{number}{1} (\bibinfo{year}{2024}), \bibinfo{pages}{1--23}.
\newblock


\bibitem[\protect\citeauthoryear{Dai, Xu, Ganesan, Alagappan, Kroth, Arpaci-Dusseau, and Arpaci-Dusseau}{Dai et~al\mbox{.}}{2020}]%
        {dai2020wisckey}
\bibfield{author}{\bibinfo{person}{Yifan Dai}, \bibinfo{person}{Yien Xu}, \bibinfo{person}{Aishwarya Ganesan}, \bibinfo{person}{Ramnatthan Alagappan}, \bibinfo{person}{Brian Kroth}, \bibinfo{person}{Andrea Arpaci-Dusseau}, {and} \bibinfo{person}{Remzi Arpaci-Dusseau}.} \bibinfo{year}{2020}\natexlab{}.
\newblock \showarticletitle{From WiscKey to Bourbon: A Learned Index for Log-Structured Merge Trees}. In \bibinfo{booktitle}{\emph{14th USENIX Symp. on Operating Systems Design and Implementation (OSDI 20)}}. \bibinfo{pages}{155--171}.
\newblock


\bibitem[\protect\citeauthoryear{Dayan, Athanassoulis, and Idreos}{Dayan et~al\mbox{.}}{2017}]%
        {dayan2017monkey}
\bibfield{author}{\bibinfo{person}{Niv Dayan}, \bibinfo{person}{Manos Athanassoulis}, {and} \bibinfo{person}{Stratos Idreos}.} \bibinfo{year}{2017}\natexlab{}.
\newblock \showarticletitle{Monkey: Optimal navigable key-value store}. In \bibinfo{booktitle}{\emph{Proceedings of the 2017 ACM International Conference on Management of Data}}. \bibinfo{pages}{79--94}.
\newblock


\bibitem[\protect\citeauthoryear{Dayan and Idreos}{Dayan and Idreos}{2018}]%
        {dostoevsky2018}
\bibfield{author}{\bibinfo{person}{Niv Dayan} {and} \bibinfo{person}{Stratos Idreos}.} \bibinfo{year}{2018}\natexlab{}.
\newblock \showarticletitle{Dostoevsky: Better Space-Time Trade-Offs for LSM-Tree Based Key-Value Stores via Adaptive Removal of Superfluous Merging}. In \bibinfo{booktitle}{\emph{Proceedings of the 2018 International Conference on Management of Data}} (Houston, TX, USA) \emph{(\bibinfo{series}{SIGMOD '18})}. \bibinfo{publisher}{Association for Computing Machinery}, \bibinfo{address}{New York, NY, USA}, \bibinfo{pages}{505–520}.
\newblock
\showISBNx{9781450347037}
\urldef\tempurl%
\url{https://doi.org/10.1145/3183713.3196927}
\showDOI{\tempurl}


\bibitem[\protect\citeauthoryear{Dayan and Idreos}{Dayan and Idreos}{2019}]%
        {dayan2019log}
\bibfield{author}{\bibinfo{person}{Niv Dayan} {and} \bibinfo{person}{Stratos Idreos}.} \bibinfo{year}{2019}\natexlab{}.
\newblock \showarticletitle{The log-structured merge-bush \& the wacky continuum}. In \bibinfo{booktitle}{\emph{Proceedings of the 2019 International Conference on Management of Data}}. \bibinfo{pages}{449--466}.
\newblock


\bibitem[\protect\citeauthoryear{Dayan and Twitto}{Dayan and Twitto}{2021}]%
        {dayan2021chucky}
\bibfield{author}{\bibinfo{person}{Niv Dayan} {and} \bibinfo{person}{Moshe Twitto}.} \bibinfo{year}{2021}\natexlab{}.
\newblock \showarticletitle{Chucky: A Succinct Cuckoo Filter for LSM-Tree}. In \bibinfo{booktitle}{\emph{Proceedings of the 2021 International Conference on Management of Data}}. \bibinfo{pages}{365--378}.
\newblock


\bibitem[\protect\citeauthoryear{Dayan, Weiss, Dashevsky, Pan, Bortnikov, and Twitto}{Dayan et~al\mbox{.}}{2022}]%
        {dayanspooky}
\bibfield{author}{\bibinfo{person}{Niv Dayan}, \bibinfo{person}{Tamar Weiss}, \bibinfo{person}{Shmuel Dashevsky}, \bibinfo{person}{Michael Pan}, \bibinfo{person}{Edward Bortnikov}, {and} \bibinfo{person}{Moshe Twitto}.} \bibinfo{year}{2022}\natexlab{}.
\newblock \showarticletitle{Spooky: granulating LSM-tree compactions correctly}.
\newblock \bibinfo{journal}{\emph{Proceedings of the VLDB Endowment}} \bibinfo{volume}{15}, \bibinfo{number}{11} (\bibinfo{year}{2022}), \bibinfo{pages}{3071--3084}.
\newblock


\bibitem[\protect\citeauthoryear{Eldawy, Hristidis, Ghosh, Saeedan, Sevim, Siddique, Singla, Sivaram, Vu, and Zhang}{Eldawy et~al\mbox{.}}{2021}]%
        {eldawy2021beast}
\bibfield{author}{\bibinfo{person}{Ahmed Eldawy}, \bibinfo{person}{Vagelis Hristidis}, \bibinfo{person}{Saheli Ghosh}, \bibinfo{person}{Majid Saeedan}, \bibinfo{person}{Akil Sevim}, \bibinfo{person}{AB Siddique}, \bibinfo{person}{Samriddhi Singla}, \bibinfo{person}{Ganesh Sivaram}, \bibinfo{person}{Tin Vu}, {and} \bibinfo{person}{Yaming Zhang}.} \bibinfo{year}{2021}\natexlab{}.
\newblock \showarticletitle{Beast: Scalable exploratory analytics on spatio-temporal data}. In \bibinfo{booktitle}{\emph{Proceedings of the 30th ACM International Conference on Information \& Knowledge Management}}. \bibinfo{pages}{3796--3807}.
\newblock


\bibitem[\protect\citeauthoryear{Facebook}{Facebook}{2025a}]%
        {rocksdb}
\bibfield{author}{\bibinfo{person}{Facebook}.} \bibinfo{year}{2025}\natexlab{a}.
\newblock \bibinfo{title}{RocksDB}.
\newblock \bibinfo{howpublished}{\url{https://github.com/facebook/rocksdb}}.
\newblock


\bibitem[\protect\citeauthoryear{Facebook}{Facebook}{2025b}]%
        {dbbench}
\bibfield{author}{\bibinfo{person}{Facebook}.} \bibinfo{year}{2025}\natexlab{b}.
\newblock \bibinfo{title}{RocksDB Benchmark}.
\newblock \bibinfo{howpublished}{\url{https://github.com/facebook/rocksdb/wiki/Benchmarking-tools}}.
\newblock


\bibitem[\protect\citeauthoryear{Facebook}{Facebook}{2025c}]%
        {tuningguide}
\bibfield{author}{\bibinfo{person}{Facebook}.} \bibinfo{year}{2025}\natexlab{c}.
\newblock \bibinfo{title}{RocksDB tuning guide}.
\newblock \bibinfo{howpublished}{\url{https://github.com/facebook/rocksdb/wiki/RocksDB-Tuning-Guide}}.
\newblock


\bibitem[\protect\citeauthoryear{Gartner, Kemper, Kossmann, and Zeller}{Gartner et~al\mbox{.}}{2001}]%
        {gartner2001efficient}
\bibfield{author}{\bibinfo{person}{A Gartner}, \bibinfo{person}{Alfons Kemper}, \bibinfo{person}{Donald Kossmann}, {and} \bibinfo{person}{Bernhard Zeller}.} \bibinfo{year}{2001}\natexlab{}.
\newblock \showarticletitle{Efficient bulk deletes in relational databases}. In \bibinfo{booktitle}{\emph{Proceedings 17th International Conference on Data Engineering}}. IEEE, \bibinfo{pages}{183--192}.
\newblock


\bibitem[\protect\citeauthoryear{Huang and Ghandeharizadeh}{Huang and Ghandeharizadeh}{2021}]%
        {huang2021nova}
\bibfield{author}{\bibinfo{person}{Haoyu Huang} {and} \bibinfo{person}{Shahram Ghandeharizadeh}.} \bibinfo{year}{2021}\natexlab{}.
\newblock \showarticletitle{Nova-LSM: a distributed, component-based LSM-tree key-value store}. In \bibinfo{booktitle}{\emph{Proceedings of the 2021 International Conference on Management of Data}}. \bibinfo{pages}{749--763}.
\newblock


\bibitem[\protect\citeauthoryear{Huynh, Chaudhari, Terzi, and Athanassoulis}{Huynh et~al\mbox{.}}{2021}]%
        {huynh2021endure}
\bibfield{author}{\bibinfo{person}{Andy Huynh}, \bibinfo{person}{Harshal Chaudhari}, \bibinfo{person}{Evimaria Terzi}, {and} \bibinfo{person}{Manos Athanassoulis}.} \bibinfo{year}{2021}\natexlab{}.
\newblock \showarticletitle{Endure: A Robust Tuning Paradigm for LSM Trees Under Workload Uncertainty}.
\newblock \bibinfo{journal}{\emph{arXiv preprint arXiv:2110.13801}} (\bibinfo{year}{2021}).
\newblock


\bibitem[\protect\citeauthoryear{Huynh, Chaudhari, Terzi, and Athanassoulis}{Huynh et~al\mbox{.}}{2024}]%
        {huynh2024towards}
\bibfield{author}{\bibinfo{person}{Andy Huynh}, \bibinfo{person}{Harshal~A Chaudhari}, \bibinfo{person}{Evimaria Terzi}, {and} \bibinfo{person}{Manos Athanassoulis}.} \bibinfo{year}{2024}\natexlab{}.
\newblock \showarticletitle{Towards flexibility and robustness of LSM trees}.
\newblock \bibinfo{journal}{\emph{The VLDB Journal}} \bibinfo{volume}{33}, \bibinfo{number}{4} (\bibinfo{year}{2024}), \bibinfo{pages}{1105--1128}.
\newblock


\bibitem[\protect\citeauthoryear{Idreos, Dayan, Qin, Akmanalp, Hilgard, Ross, Lennon, Jain, Gupta, Li, and Zhu}{Idreos et~al\mbox{.}}{2019}]%
        {idreos2019designcontinuum}
\bibfield{author}{\bibinfo{person}{Stratos Idreos}, \bibinfo{person}{Niv Dayan}, \bibinfo{person}{Wilson Qin}, \bibinfo{person}{Mali Akmanalp}, \bibinfo{person}{Sophie Hilgard}, \bibinfo{person}{Andrew~Slavin Ross}, \bibinfo{person}{James Lennon}, \bibinfo{person}{Varun Jain}, \bibinfo{person}{Harshita Gupta}, \bibinfo{person}{David Li}, {and} \bibinfo{person}{Zichen Zhu}.} \bibinfo{year}{2019}\natexlab{}.
\newblock \showarticletitle{Design Continuums and the Path Toward Self-Designing Key-Value Stores that Know and Learn}. In \bibinfo{booktitle}{\emph{Proceedings of the Biennial Conference on Innovative Data Systems Research (CIDR)}}.
\newblock


\bibitem[\protect\citeauthoryear{Influxdata}{Influxdata}{2025a}]%
        {influxdb}
\bibfield{author}{\bibinfo{person}{Influxdata}.} \bibinfo{year}{2025}\natexlab{a}.
\newblock \bibinfo{title}{InfluxDB}.
\newblock \bibinfo{howpublished}{\url{https://https://www.influxdata.com/}}.
\newblock


\bibitem[\protect\citeauthoryear{Influxdata}{Influxdata}{2025b}]%
        {influxdelete}
\bibfield{author}{\bibinfo{person}{Influxdata}.} \bibinfo{year}{2025}\natexlab{b}.
\newblock \bibinfo{title}{InfluxDB Delete Data}.
\newblock \bibinfo{howpublished}{\url{https://docs.influxdata.com/influxdb/v2/write-data/delete-data/}}.
\newblock


\bibitem[\protect\citeauthoryear{Kim, Behm, Blow, Borkar, Bu, Carey, Hubail, Jahangiri, Jia, Li, et~al\mbox{.}}{Kim et~al\mbox{.}}{2020}]%
        {kim2020robust}
\bibfield{author}{\bibinfo{person}{Taewoo Kim}, \bibinfo{person}{Alexander Behm}, \bibinfo{person}{Michael Blow}, \bibinfo{person}{Vinayak Borkar}, \bibinfo{person}{Yingyi Bu}, \bibinfo{person}{Michael~J Carey}, \bibinfo{person}{Murtadha Hubail}, \bibinfo{person}{Shiva Jahangiri}, \bibinfo{person}{Jianfeng Jia}, \bibinfo{person}{Chen Li}, {et~al\mbox{.}}} \bibinfo{year}{2020}\natexlab{}.
\newblock \showarticletitle{Robust and efficient memory management in Apache AsterixDB}.
\newblock \bibinfo{journal}{\emph{Software: Practice and Experience}} \bibinfo{volume}{50}, \bibinfo{number}{7} (\bibinfo{year}{2020}), \bibinfo{pages}{1114--1151}.
\newblock


\bibitem[\protect\citeauthoryear{Knorr, Lemaire, Lim, Luo, Zhang, Idreos, and Mitzenmacher}{Knorr et~al\mbox{.}}{2022}]%
        {knorr2022proteus}
\bibfield{author}{\bibinfo{person}{Eric~R Knorr}, \bibinfo{person}{Baptiste Lemaire}, \bibinfo{person}{Andrew Lim}, \bibinfo{person}{Siqiang Luo}, \bibinfo{person}{Huanchen Zhang}, \bibinfo{person}{Stratos Idreos}, {and} \bibinfo{person}{Michael Mitzenmacher}.} \bibinfo{year}{2022}\natexlab{}.
\newblock \showarticletitle{Proteus: A Self-Designing Range Filter}. In \bibinfo{booktitle}{\emph{Proceedings of the 2022 International Conference on Management of Data}}. \bibinfo{pages}{1670--1684}.
\newblock


\bibitem[\protect\citeauthoryear{Labs}{Labs}{2025}]%
        {cockroachdb}
\bibfield{author}{\bibinfo{person}{Cockroach Labs}.} \bibinfo{year}{2025}\natexlab{}.
\newblock \bibinfo{title}{CockroachDB}.
\newblock \bibinfo{howpublished}{\url{https://github.com/cockroachdb/cockroach}}.
\newblock


\bibitem[\protect\citeauthoryear{Lakshman and Malik}{Lakshman and Malik}{2010}]%
        {lakshman2010cassandra}
\bibfield{author}{\bibinfo{person}{Avinash Lakshman} {and} \bibinfo{person}{Prashant Malik}.} \bibinfo{year}{2010}\natexlab{}.
\newblock \showarticletitle{Cassandra: a decentralized structured storage system}.
\newblock \bibinfo{journal}{\emph{ACM SIGOPS Operating Systems Review}} \bibinfo{volume}{44}, \bibinfo{number}{2} (\bibinfo{year}{2010}), \bibinfo{pages}{35--40}.
\newblock


\bibitem[\protect\citeauthoryear{Li, Chen, Dai, Xie, Luo, Gu, Yang, and Chen}{Li et~al\mbox{.}}{2022}]%
        {li2022seesaw}
\bibfield{author}{\bibinfo{person}{Meng Li}, \bibinfo{person}{Deyi Chen}, \bibinfo{person}{Haipeng Dai}, \bibinfo{person}{Rongbiao Xie}, \bibinfo{person}{Siqiang Luo}, \bibinfo{person}{Rong Gu}, \bibinfo{person}{Tong Yang}, {and} \bibinfo{person}{Guihai Chen}.} \bibinfo{year}{2022}\natexlab{}.
\newblock \showarticletitle{Seesaw Counting Filter: An Efficient Guardian for Vulnerable Negative Keys During Dynamic Filtering}. In \bibinfo{booktitle}{\emph{Proceedings of the ACM Web Conference 2022}}. \bibinfo{pages}{2759--2767}.
\newblock


\bibitem[\protect\citeauthoryear{Lilja, Saikkonen, Sippu, and Soisalon-Soininen}{Lilja et~al\mbox{.}}{2006}]%
        {lilja2006online}
\bibfield{author}{\bibinfo{person}{Timo Lilja}, \bibinfo{person}{Riku Saikkonen}, \bibinfo{person}{Seppo Sippu}, {and} \bibinfo{person}{Eljas Soisalon-Soininen}.} \bibinfo{year}{2006}\natexlab{}.
\newblock \showarticletitle{Online bulk deletion}. In \bibinfo{booktitle}{\emph{2007 IEEE 23rd International Conference on Data Engineering}}. IEEE, \bibinfo{pages}{956--965}.
\newblock


\bibitem[\protect\citeauthoryear{Liu, Wang, Mo, and Luo}{Liu et~al\mbox{.}}{2024}]%
        {liu2024structural}
\bibfield{author}{\bibinfo{person}{Junfeng Liu}, \bibinfo{person}{Fan Wang}, \bibinfo{person}{Dingheng Mo}, {and} \bibinfo{person}{Siqiang Luo}.} \bibinfo{year}{2024}\natexlab{}.
\newblock \showarticletitle{Structural Designs Meet Optimality: Exploring Optimized LSM-tree Structures in A Colossal Configuration Space}.
\newblock \bibinfo{journal}{\emph{PACMMOD}} \bibinfo{volume}{2}, \bibinfo{number}{3} (\bibinfo{year}{2024}), \bibinfo{pages}{1--26}.
\newblock


\bibitem[\protect\citeauthoryear{Lu, Pillai, Gopalakrishnan, Arpaci-Dusseau, and Arpaci-Dusseau}{Lu et~al\mbox{.}}{2017}]%
        {lu2017wisckey}
\bibfield{author}{\bibinfo{person}{Lanyue Lu}, \bibinfo{person}{Thanumalayan~Sankaranarayana Pillai}, \bibinfo{person}{Hariharan Gopalakrishnan}, \bibinfo{person}{Andrea~C Arpaci-Dusseau}, {and} \bibinfo{person}{Remzi~H Arpaci-Dusseau}.} \bibinfo{year}{2017}\natexlab{}.
\newblock \showarticletitle{Wisckey: Separating keys from values in ssd-conscious storage}.
\newblock \bibinfo{journal}{\emph{ACM Transactions on Storage (TOS)}} \bibinfo{volume}{13}, \bibinfo{number}{1} (\bibinfo{year}{2017}), \bibinfo{pages}{1--28}.
\newblock


\bibitem[\protect\citeauthoryear{Luo and Carey}{Luo and Carey}{2020}]%
        {luo2020breaking}
\bibfield{author}{\bibinfo{person}{Chen Luo} {and} \bibinfo{person}{Michael~J Carey}.} \bibinfo{year}{2020}\natexlab{}.
\newblock \showarticletitle{Breaking down memory walls: adaptive memory management in LSM-based storage systems}.
\newblock \bibinfo{journal}{\emph{Proceedings of the VLDB Endowment}} \bibinfo{volume}{14}, \bibinfo{number}{3} (\bibinfo{year}{2020}), \bibinfo{pages}{241--254}.
\newblock


\bibitem[\protect\citeauthoryear{Luo, Chatterjee, Ketsetsidis, Dayan, Qin, and Idreos}{Luo et~al\mbox{.}}{2020}]%
        {luo2020rosetta}
\bibfield{author}{\bibinfo{person}{Siqiang Luo}, \bibinfo{person}{Subarna Chatterjee}, \bibinfo{person}{Rafael Ketsetsidis}, \bibinfo{person}{Niv Dayan}, \bibinfo{person}{Wilson Qin}, {and} \bibinfo{person}{Stratos Idreos}.} \bibinfo{year}{2020}\natexlab{}.
\newblock \showarticletitle{Rosetta: A robust space-time optimized range filter for key-value stores}. In \bibinfo{booktitle}{\emph{Proceedings of the 2020 ACM SIGMOD International Conference on Management of Data}}. \bibinfo{pages}{2071--2086}.
\newblock


\bibitem[\protect\citeauthoryear{Luo, Kao, Li, Hu, Cheng, and Zheng}{Luo et~al\mbox{.}}{2018}]%
        {luo2018toain}
\bibfield{author}{\bibinfo{person}{Siqiang Luo}, \bibinfo{person}{Ben Kao}, \bibinfo{person}{Guoliang Li}, \bibinfo{person}{Jiafeng Hu}, \bibinfo{person}{Reynold Cheng}, {and} \bibinfo{person}{Yudian Zheng}.} \bibinfo{year}{2018}\natexlab{}.
\newblock \showarticletitle{Toain: a throughput optimizing adaptive index for answering dynamic k nn queries on road networks}.
\newblock \bibinfo{journal}{\emph{Proceedings of the VLDB Endowment}} \bibinfo{volume}{11}, \bibinfo{number}{5} (\bibinfo{year}{2018}), \bibinfo{pages}{594--606}.
\newblock


\bibitem[\protect\citeauthoryear{Lv, Li, Xu, Gao, Yang, Wang, and Xue}{Lv et~al\mbox{.}}{2025}]%
        {lv2025rethinking}
\bibfield{author}{\bibinfo{person}{Yina Lv}, \bibinfo{person}{Qiao Li}, \bibinfo{person}{Quanqing Xu}, \bibinfo{person}{Congming Gao}, \bibinfo{person}{Chuanhui Yang}, \bibinfo{person}{Xiaoli Wang}, {and} \bibinfo{person}{Chun~Jason Xue}.} \bibinfo{year}{2025}\natexlab{}.
\newblock \showarticletitle{Rethinking LSM-tree based Key-Value Stores: A Survey}.
\newblock \bibinfo{journal}{\emph{arXiv preprint arXiv:2507.09642}} (\bibinfo{year}{2025}).
\newblock


\bibitem[\protect\citeauthoryear{Mao, Qader, and Hristidis}{Mao et~al\mbox{.}}{2020}]%
        {mao2020comprehensive}
\bibfield{author}{\bibinfo{person}{Qizhong Mao}, \bibinfo{person}{Mohiuddin~Abdul Qader}, {and} \bibinfo{person}{Vagelis Hristidis}.} \bibinfo{year}{2020}\natexlab{}.
\newblock \showarticletitle{Comprehensive comparison of LSM architectures for spatial data}. In \bibinfo{booktitle}{\emph{2020 IEEE International Conference on Big Data (Big Data)}}. IEEE, \bibinfo{pages}{455--460}.
\newblock


\bibitem[\protect\citeauthoryear{Mo, Chen, Luo, and Shan}{Mo et~al\mbox{.}}{2023}]%
        {mo2023learning}
\bibfield{author}{\bibinfo{person}{Dingheng Mo}, \bibinfo{person}{Fanchao Chen}, \bibinfo{person}{Siqiang Luo}, {and} \bibinfo{person}{Caihua Shan}.} \bibinfo{year}{2023}\natexlab{}.
\newblock \showarticletitle{Learning to Optimize LSM-trees: Towards A Reinforcement Learning based Key-Value Store for Dynamic Workloads}.
\newblock \bibinfo{journal}{\emph{arXiv preprint arXiv:2308.07013}} (\bibinfo{year}{2023}).
\newblock


\bibitem[\protect\citeauthoryear{Mo, Liu, Wang, and Luo}{Mo et~al\mbox{.}}{2025a}]%
        {mo2025aster}
\bibfield{author}{\bibinfo{person}{Dingheng Mo}, \bibinfo{person}{Junfeng Liu}, \bibinfo{person}{Fan Wang}, {and} \bibinfo{person}{Siqiang Luo}.} \bibinfo{year}{2025}\natexlab{a}.
\newblock \showarticletitle{Aster: Enhancing LSM-structures for Scalable Graph Database}.
\newblock \bibinfo{journal}{\emph{Proceedings of the ACM on Management of Data}} \bibinfo{volume}{3}, \bibinfo{number}{1} (\bibinfo{year}{2025}), \bibinfo{pages}{1--26}.
\newblock


\bibitem[\protect\citeauthoryear{Mo, Luo, and Idreos}{Mo et~al\mbox{.}}{2025b}]%
        {mo2025grow}
\bibfield{author}{\bibinfo{person}{Dingheng Mo}, \bibinfo{person}{Siqiang Luo}, {and} \bibinfo{person}{Stratos Idreos}.} \bibinfo{year}{2025}\natexlab{b}.
\newblock \showarticletitle{How to Grow an LSM-tree? Towards Bridging the Gap Between Theory and Practice}.
\newblock \bibinfo{journal}{\emph{Proceedings of the ACM on Management of Data}} \bibinfo{volume}{3}, \bibinfo{number}{3} (\bibinfo{year}{2025}), \bibinfo{pages}{1--25}.
\newblock


\bibitem[\protect\citeauthoryear{M{\"o}{\ss}ner, Riegger, Bernhardt, and Petrov}{M{\"o}{\ss}ner et~al\mbox{.}}{2023}]%
        {mossner2023bloomrf}
\bibfield{author}{\bibinfo{person}{Bernhard M{\"o}{\ss}ner}, \bibinfo{person}{Christian Riegger}, \bibinfo{person}{Arthur Bernhardt}, {and} \bibinfo{person}{Ilia Petrov}.} \bibinfo{year}{2023}\natexlab{}.
\newblock \showarticletitle{bloomRF: On performing range-queries in Bloom-Filters with piecewise-monotone hash functions and prefix hashing}. In \bibinfo{booktitle}{\emph{Advances in database technology: Proceedings of the 26th International Conference on Extending database Technology (EDBT), 28th March-31st March 2023, Ioannina, Greece}}, Vol.~\bibinfo{volume}{26}. Open Proceedings. org, Univ. of Konstanz, \bibinfo{pages}{131--143}.
\newblock


\bibitem[\protect\citeauthoryear{Nie, Lei, Li, Niu, Liu, and Wu}{Nie et~al\mbox{.}}{2024}]%
        {nie2024zone}
\bibfield{author}{\bibinfo{person}{Shiqiang Nie}, \bibinfo{person}{Tong Lei}, \bibinfo{person}{Menghan Li}, \bibinfo{person}{Jie Niu}, \bibinfo{person}{Song Liu}, {and} \bibinfo{person}{Weiguo Wu}.} \bibinfo{year}{2024}\natexlab{}.
\newblock \showarticletitle{Zone-Aware Persistent Deletion for Key-Value Store Engine}. In \bibinfo{booktitle}{\emph{2024 13th Non-Volatile Memory Systems and Applications Symposium (NVMSA)}}. IEEE, \bibinfo{pages}{1--6}.
\newblock


\bibitem[\protect\citeauthoryear{Pan, Yue, and Xiong}{Pan et~al\mbox{.}}{2017}]%
        {pan2017dcompaction}
\bibfield{author}{\bibinfo{person}{Fengfeng Pan}, \bibinfo{person}{Yinliang Yue}, {and} \bibinfo{person}{Jin Xiong}.} \bibinfo{year}{2017}\natexlab{}.
\newblock \showarticletitle{dCompaction: Delayed compaction for the LSM-tree}.
\newblock \bibinfo{journal}{\emph{International Journal of Parallel Programming}} \bibinfo{volume}{45}, \bibinfo{number}{6} (\bibinfo{year}{2017}), \bibinfo{pages}{1310--1325}.
\newblock


\bibitem[\protect\citeauthoryear{Pang, Lu, Chen, Wang, Xu, and Wu}{Pang et~al\mbox{.}}{2021}]%
        {pang2021arkdb}
\bibfield{author}{\bibinfo{person}{Zhu Pang}, \bibinfo{person}{Qingda Lu}, \bibinfo{person}{Shuo Chen}, \bibinfo{person}{Rui Wang}, \bibinfo{person}{Yikang Xu}, {and} \bibinfo{person}{Jiesheng Wu}.} \bibinfo{year}{2021}\natexlab{}.
\newblock \showarticletitle{ArkDB: a key-value engine for scalable cloud storage services}. In \bibinfo{booktitle}{\emph{Proceedings of the 2021 International Conference on Management of Data}}. \bibinfo{pages}{2570--2583}.
\newblock


\bibitem[\protect\citeauthoryear{Raju, Kadekodi, Chidambaram, and Abraham}{Raju et~al\mbox{.}}{2017}]%
        {raju2017pebblesdb}
\bibfield{author}{\bibinfo{person}{Pandian Raju}, \bibinfo{person}{Rohan Kadekodi}, \bibinfo{person}{Vijay Chidambaram}, {and} \bibinfo{person}{Ittai Abraham}.} \bibinfo{year}{2017}\natexlab{}.
\newblock \showarticletitle{Pebblesdb: Building key-value stores using fragmented log-structured merge trees}. In \bibinfo{booktitle}{\emph{Proceedings of the 26th Symp. on Operating Systems Principles}}. \bibinfo{pages}{497--514}.
\newblock


\bibitem[\protect\citeauthoryear{Sarkar}{Sarkar}{2021}]%
        {lethecode}
\bibfield{author}{\bibinfo{person}{Sarkar}.} \bibinfo{year}{2021}\natexlab{}.
\newblock \bibinfo{title}{Lethe}.
\newblock \bibinfo{howpublished}{\url{https://github.com/BU-DiSC/lethe-codebase/tree/main}}.
\newblock


\bibitem[\protect\citeauthoryear{Sarkar, Papon, Staratzis, and Athanassoulis}{Sarkar et~al\mbox{.}}{2020}]%
        {sarkar2020lethe}
\bibfield{author}{\bibinfo{person}{Subhadeep Sarkar}, \bibinfo{person}{Tarikul~Islam Papon}, \bibinfo{person}{Dimitris Staratzis}, {and} \bibinfo{person}{Manos Athanassoulis}.} \bibinfo{year}{2020}\natexlab{}.
\newblock \showarticletitle{Lethe: A tunable delete-aware LSM engine}. In \bibinfo{booktitle}{\emph{Proceedings of the 2020 ACM SIGMOD International Conference on Management of Data}}. \bibinfo{pages}{893--908}.
\newblock


\bibitem[\protect\citeauthoryear{Shin, Wang, and Aref}{Shin et~al\mbox{.}}{2021}]%
        {shin2021lsm}
\bibfield{author}{\bibinfo{person}{Jaewoo Shin}, \bibinfo{person}{Jianguo Wang}, {and} \bibinfo{person}{Walid~G Aref}.} \bibinfo{year}{2021}\natexlab{}.
\newblock \showarticletitle{The LSM rum-tree: A log structured merge r-tree for update-intensive spatial workloads}. In \bibinfo{booktitle}{\emph{2021 IEEE 37th international conference on data engineering (ICDE)}}. IEEE, \bibinfo{pages}{2285--2290}.
\newblock


\bibitem[\protect\citeauthoryear{Systems}{Systems}{2025}]%
        {scylladb}
\bibfield{author}{\bibinfo{person}{Cloudius Systems}.} \bibinfo{year}{2025}\natexlab{}.
\newblock \bibinfo{title}{ScyllaDB}.
\newblock \bibinfo{howpublished}{\url{https://www.scylladb.com/}}.
\newblock


\bibitem[\protect\citeauthoryear{Thonangi and Yang}{Thonangi and Yang}{2017}]%
        {thonangi2017log}
\bibfield{author}{\bibinfo{person}{Risi Thonangi} {and} \bibinfo{person}{Jun Yang}.} \bibinfo{year}{2017}\natexlab{}.
\newblock \showarticletitle{On log-structured merge for solid-state drives}. In \bibinfo{booktitle}{\emph{2017 IEEE 33rd International Conference on Data Engineering (ICDE)}}. IEEE, \bibinfo{pages}{683--694}.
\newblock


\bibitem[\protect\citeauthoryear{Vaidya, Chatterjee, Knorr, Mitzenmacher, Idreos, and Kraska}{Vaidya et~al\mbox{.}}{2022}]%
        {vaidya2022snarf}
\bibfield{author}{\bibinfo{person}{Kapil Vaidya}, \bibinfo{person}{Subarna Chatterjee}, \bibinfo{person}{Eric Knorr}, \bibinfo{person}{Michael Mitzenmacher}, \bibinfo{person}{Stratos Idreos}, {and} \bibinfo{person}{Tim Kraska}.} \bibinfo{year}{2022}\natexlab{}.
\newblock \showarticletitle{SNARF: a learning-enhanced range filter}.
\newblock \bibinfo{journal}{\emph{Proceedings of the VLDB Endowment}} \bibinfo{volume}{15}, \bibinfo{number}{8} (\bibinfo{year}{2022}), \bibinfo{pages}{1632--1644}.
\newblock


\bibitem[\protect\citeauthoryear{Vu, Eldawy, Hristidis, and Tsotras}{Vu et~al\mbox{.}}{2021}]%
        {vu2021incremental}
\bibfield{author}{\bibinfo{person}{Tin Vu}, \bibinfo{person}{Ahmed Eldawy}, \bibinfo{person}{Vagelis Hristidis}, {and} \bibinfo{person}{Vassilis Tsotras}.} \bibinfo{year}{2021}\natexlab{}.
\newblock \showarticletitle{Incremental partitioning for efficient spatial data analytics}.
\newblock \bibinfo{journal}{\emph{Proceedings of the VLDB Endowment}} \bibinfo{volume}{15}, \bibinfo{number}{3} (\bibinfo{year}{2021}), \bibinfo{pages}{713--726}.
\newblock


\bibitem[\protect\citeauthoryear{Wang, Guo, Yang, and Zhang}{Wang et~al\mbox{.}}{2024}]%
        {wang2024grf}
\bibfield{author}{\bibinfo{person}{Hengrui Wang}, \bibinfo{person}{Te Guo}, \bibinfo{person}{Junzhao Yang}, {and} \bibinfo{person}{Huanchen Zhang}.} \bibinfo{year}{2024}\natexlab{}.
\newblock \showarticletitle{GRF: A Global Range Filter for LSM-Trees with Shape Encoding}.
\newblock \bibinfo{journal}{\emph{PACMMOD}} \bibinfo{volume}{2}, \bibinfo{number}{3} (\bibinfo{year}{2024}), \bibinfo{pages}{1--27}.
\newblock


\bibitem[\protect\citeauthoryear{Wang and Shao}{Wang and Shao}{2023}]%
        {wang2023mirrorkv}
\bibfield{author}{\bibinfo{person}{Zhiqi Wang} {and} \bibinfo{person}{Zili Shao}.} \bibinfo{year}{2023}\natexlab{}.
\newblock \showarticletitle{MirrorKV: An Efficient Key-Value Store on Hybrid Cloud Storage with Balanced Performance of Compaction and Querying}.
\newblock \bibinfo{journal}{\emph{Proceedings of the ACM on Management of Data}} \bibinfo{volume}{1}, \bibinfo{number}{4} (\bibinfo{year}{2023}), \bibinfo{pages}{1--27}.
\newblock


\bibitem[\protect\citeauthoryear{Wang, Zhong, Guo, Wu, Li, Yang, Tu, Zhang, and Cui}{Wang et~al\mbox{.}}{2023}]%
        {wang2023rencoder}
\bibfield{author}{\bibinfo{person}{Ziwei Wang}, \bibinfo{person}{Zheng Zhong}, \bibinfo{person}{Jiarui Guo}, \bibinfo{person}{Yuhan Wu}, \bibinfo{person}{Haoyu Li}, \bibinfo{person}{Tong Yang}, \bibinfo{person}{Yaofeng Tu}, \bibinfo{person}{Huanchen Zhang}, {and} \bibinfo{person}{Bin Cui}.} \bibinfo{year}{2023}\natexlab{}.
\newblock \showarticletitle{Rencoder: A space-time efficient range filter with local encoder}. In \bibinfo{booktitle}{\emph{2023 IEEE 39th International Conference on Data Engineering (ICDE)}}. IEEE, \bibinfo{pages}{2036--2049}.
\newblock


\bibitem[\protect\citeauthoryear{Wu, Xu, Shao, and Jiang}{Wu et~al\mbox{.}}{2015}]%
        {wu2015lsmtrie}
\bibfield{author}{\bibinfo{person}{Xingbo Wu}, \bibinfo{person}{Yuehai Xu}, \bibinfo{person}{Zili Shao}, {and} \bibinfo{person}{Song Jiang}.} \bibinfo{year}{2015}\natexlab{}.
\newblock \showarticletitle{LSM-trie: An LSM-tree-based Ultra-Large Key-Value Store for Small Data Items}. In \bibinfo{booktitle}{\emph{2015 USENIX Annual Technical Conference (USENIX ATC 15)}}. \bibinfo{pages}{71--82}.
\newblock


\bibitem[\protect\citeauthoryear{Yahoo!}{Yahoo!}{2010}]%
        {ycsb}
\bibfield{author}{\bibinfo{person}{Yahoo!}} \bibinfo{year}{2010}\natexlab{}.
\newblock \bibinfo{title}{Yahoo! Cloud Serving Benchmark(YCSB)}.
\newblock \bibinfo{howpublished}{\url{https://github.com/brianfrankcooper/YCSB}}.
\newblock


\bibitem[\protect\citeauthoryear{Yao, Wan, Huang, He, Gui, Wu, and Xie}{Yao et~al\mbox{.}}{2017}]%
        {yao2017light}
\bibfield{author}{\bibinfo{person}{Ting Yao}, \bibinfo{person}{Jiguang Wan}, \bibinfo{person}{Ping Huang}, \bibinfo{person}{Xubin He}, \bibinfo{person}{Qingxin Gui}, \bibinfo{person}{Fei Wu}, {and} \bibinfo{person}{Changsheng Xie}.} \bibinfo{year}{2017}\natexlab{}.
\newblock \showarticletitle{A light-weight compaction tree to reduce I/O amplification toward efficient key-value stores}. In \bibinfo{booktitle}{\emph{Proc. 33rd Int. Conf. Massive Storage Syst. Technol.(MSST)}}. \bibinfo{pages}{1--13}.
\newblock


\bibitem[\protect\citeauthoryear{Yu, Gong, Tao, Shen, Zhang, Yu, Liu, Zhang, Li, Luo, et~al\mbox{.}}{Yu et~al\mbox{.}}{2024}]%
        {yu2024lsmgraph}
\bibfield{author}{\bibinfo{person}{Song Yu}, \bibinfo{person}{Shufeng Gong}, \bibinfo{person}{Qian Tao}, \bibinfo{person}{Sijie Shen}, \bibinfo{person}{Yanfeng Zhang}, \bibinfo{person}{Wenyuan Yu}, \bibinfo{person}{Pengxi Liu}, \bibinfo{person}{Zhixin Zhang}, \bibinfo{person}{Hongfu Li}, \bibinfo{person}{Xiaojian Luo}, {et~al\mbox{.}}} \bibinfo{year}{2024}\natexlab{}.
\newblock \showarticletitle{LSMGraph: A High-Performance Dynamic Graph Storage System with Multi-Level CSR}.
\newblock \bibinfo{journal}{\emph{Proceedings of the ACM on Management of Data}} \bibinfo{volume}{2}, \bibinfo{number}{6} (\bibinfo{year}{2024}), \bibinfo{pages}{1--28}.
\newblock


\bibitem[\protect\citeauthoryear{{Yu, Geoffrey X and Markakis, Markos and Kipf, Andreas and Larson, Per-{\AA}ke and Minhas, Umar Farooq and Kraska, Tim}}{{Yu, Geoffrey X and Markakis, Markos and Kipf, Andreas and Larson, Per-{\AA}ke and Minhas, Umar Farooq and Kraska, Tim}}{2022}]%
        {yu2022treeline}
\bibfield{author}{\bibinfo{person}{{Yu, Geoffrey X and Markakis, Markos and Kipf, Andreas and Larson, Per-{\AA}ke and Minhas, Umar Farooq and Kraska, Tim}}.} \bibinfo{year}{2022}\natexlab{}.
\newblock \showarticletitle{{TreeLine: an update-in-place key-value store for modern storage}}.
\newblock \bibinfo{journal}{\emph{Proceedings of the VLDB Endowment}} \bibinfo{volume}{16}, \bibinfo{number}{1} (\bibinfo{year}{2022}), \bibinfo{pages}{99--112}.
\newblock


\bibitem[\protect\citeauthoryear{Zhang, Lim, Leis, Andersen, Kaminsky, Keeton, and Pavlo}{Zhang et~al\mbox{.}}{2018b}]%
        {zhang2018surf}
\bibfield{author}{\bibinfo{person}{Huanchen Zhang}, \bibinfo{person}{Hyeontaek Lim}, \bibinfo{person}{Viktor Leis}, \bibinfo{person}{David~G Andersen}, \bibinfo{person}{Michael Kaminsky}, \bibinfo{person}{Kimberly Keeton}, {and} \bibinfo{person}{Andrew Pavlo}.} \bibinfo{year}{2018}\natexlab{b}.
\newblock \showarticletitle{Surf: Practical range query filtering with fast succinct tries}. In \bibinfo{booktitle}{\emph{Proceedings of the 2018 International Conference on Management of Data}}. \bibinfo{pages}{323--336}.
\newblock


\bibitem[\protect\citeauthoryear{Zhang, Tan, Cai, Wang, Li, and Sun}{Zhang et~al\mbox{.}}{2022b}]%
        {zhang2022sa}
\bibfield{author}{\bibinfo{person}{Teng Zhang}, \bibinfo{person}{Jian Tan}, \bibinfo{person}{Xin Cai}, \bibinfo{person}{Jianying Wang}, \bibinfo{person}{Feifei Li}, {and} \bibinfo{person}{Jianling Sun}.} \bibinfo{year}{2022}\natexlab{b}.
\newblock \showarticletitle{SA-LSM: optimize data layout for LSM-tree based storage using survival analysis}.
\newblock \bibinfo{journal}{\emph{Proceedings of the VLDB Endowment}} \bibinfo{volume}{15}, \bibinfo{number}{10} (\bibinfo{year}{2022}), \bibinfo{pages}{2161--2174}.
\newblock


\bibitem[\protect\citeauthoryear{Zhang, Mao, Eldawy, Hristidis, and Sun}{Zhang et~al\mbox{.}}{2022a}]%
        {zhang2022bi}
\bibfield{author}{\bibinfo{person}{Xin Zhang}, \bibinfo{person}{Qizhong Mao}, \bibinfo{person}{Ahmed Eldawy}, \bibinfo{person}{Vagelis Hristidis}, {and} \bibinfo{person}{Yihan Sun}.} \bibinfo{year}{2022}\natexlab{a}.
\newblock \showarticletitle{Bi-directional Log-Structured Merge Tree}. In \bibinfo{booktitle}{\emph{Proceedings of the 34th International Conference on Scientific and Statistical Database Management}}. \bibinfo{pages}{1--4}.
\newblock


\bibitem[\protect\citeauthoryear{Zhang, Li, Guo, Li, and Xu}{Zhang et~al\mbox{.}}{2018a}]%
        {zhang2018elasticbf}
\bibfield{author}{\bibinfo{person}{Yueming Zhang}, \bibinfo{person}{Yongkun Li}, \bibinfo{person}{Fan Guo}, \bibinfo{person}{Cheng Li}, {and} \bibinfo{person}{Yinlong Xu}.} \bibinfo{year}{2018}\natexlab{a}.
\newblock \showarticletitle{ElasticBF: Fine-grained and Elastic Bloom Filter Towards Efficient Read for LSM-tree-based KV Stores}. In \bibinfo{booktitle}{\emph{10th USENIX Workshop on Hot Topics in Storage and File Systems (HotStorage 18)}}.
\newblock


\bibitem[\protect\citeauthoryear{Zhu, Mun, Raman, and Athanassoulis}{Zhu et~al\mbox{.}}{2021}]%
        {zhu2021reducing}
\bibfield{author}{\bibinfo{person}{Zichen Zhu}, \bibinfo{person}{Ju~Hyoung Mun}, \bibinfo{person}{Aneesh Raman}, {and} \bibinfo{person}{Manos Athanassoulis}.} \bibinfo{year}{2021}\natexlab{}.
\newblock \showarticletitle{Reducing bloom filter cpu overhead in lsm-trees on modern storage devices}. In \bibinfo{booktitle}{\emph{Proceedings of the 17th International Workshop on Data Management on New Hardware (DaMoN 2021)}}. \bibinfo{pages}{1--10}.
\newblock


\end{thebibliography}

\end{document}